\documentclass[10pt, twocolumn]{IEEEtran}
\usepackage[a4paper]{geometry}
\usepackage{paralist}

\usepackage{amsmath}
\usepackage{amsthm}
\usepackage{amssymb}
\usepackage{verbatim}
\usepackage{algorithm}
\usepackage{algorithmic}
\usepackage{graphicx}
\usepackage{epstopdf}  
\usepackage{paralist}
\usepackage{color}
\usepackage{url}
\usepackage{subfig}

\newtheorem{lemma}{Lemma}
\newtheorem{theorem}{Theorem}
\newtheorem{definition}{Definition}

\newtheorem{example}{Example}
\floatname{algorithm}{Algorithm}

\begin{document}
%
\title{\huge Joint Rate Adaptation and Medium Access in Wireless LANs: a Non-cooperative Game Theoretic Perspective}
\newcounter{one}
\setcounter{one}{1}
\newcounter{two}
\setcounter{two}{2}

\author{\qquad Lin Chen$^\fnsymbol{one}$ \qquad\qquad\qquad\qquad\qquad\qquad\qquad Athanasios V. Vasilakos$^\fnsymbol{two}$ \\
\parbox{0.45\textwidth}{\centering $^\fnsymbol{one}$Laboratoire de Recherche en Informatique (LRI) \\
University of Paris-Sud, 91405 Orsay, France \\
chen@lri.fr}
\hfill
\parbox{0.52\textwidth}{\centering $^\fnsymbol{two}$Dept. Computer and Telecommunications Engineering \\
Univ. Western Macedonia, Greece \\
vasilako@ath.forthnet.gr}
}

\addtolength{\abovedisplayskip}{-1ex}
\addtolength{\belowdisplayskip}{-1ex}
\addtolength{\abovedisplayshortskip}{-0.75ex}
\addtolength{\belowdisplayshortskip}{-0.75ex}


\maketitle

\begin{abstract}
Wireless local area networks (WLANs) based on IEEE 802.11 standards are becoming ubiquitous today and typically support multiple data rates. In such multi-rate WLANs, distributed medium access and rate adaptation are two key elements to achieve efficient radio resource utilization, especially in non-cooperative environments. In this paper, we present an analytical study on the non-cooperative multi-rate WLANs composed of selfish users jointly adjusting their data rate and contention window size at the medium access level to maximize their own throughput, irrespective of the impact of their selfish behaviors on overall system performance.
Specifically, we develop an adapted Tit-For-Tat (TFT) strategy to guide the system to an efficient equilibrium in non-cooperative environments. We model the interactions among selfish users under the adapted TFT framework as a non-cooperative joint medium access and rate adaptation game. A systematic analysis is conducted on the structural properties of the game to provide insights on the interaction between rate adaptation and 802.11 medium access control in a competitive setting. We show that the game has multiple equilibria, which, after the equilibrium refinement process that we develop, reduce to a unique efficient equilibrium. We further develop a distributed algorithm to achieve this equilibrium and demonstrate that the equilibrium achieves the performance very close to the system optimum in a social perspective.
\end{abstract}

\section{Introduction}
\label{sec:intro}

\subsection{General Context}

Wireless local area networks (WLANs) based on IEEE 802.11 standards are becoming ubiquitous today and typically support multiple data transmission rates, e.g., four in 802.11b, eight in 802.11a/g and 20 in 802.11n. In such multi-rate WLANs, distributed medium access and rate adaptation are two key elements to achieve efficient radio resource utilization. In this paper, we focus on the competitive scenario where selfish users jointly adjust their data rate and contention window (CW) size at the medium access level to maximize their own throughput, irrespective of the impact of their selfish behaviors on overall system performance. Our focus on the non-cooperative scenario is motivated by the following two observations:

Today, network adapters are highly programmable, making a selfish user extremely easy to tamper wireless interface to maximize its own benefit, especially in environments without central administration such as public hot-spots or WLANs operated by different enterprisers;

WLANs are by nature distributed environments lack of coordination and sophisticated network feedback. In such environments, non-cooperative selfish behaviors maximizing local utilities are much more robust and scalable than any centralized cooperative control, which is very expensive or even impossible to implement.

\subsection{Related Work}

The existing literature on the non-cooperative medium access and rate control in WLANs can be naturally categorized in the following two domains:

\textit{Non-cooperative medium access:} Recent studies \cite{Cagalj05}, \cite{Konorski06} have studied some undesirable effects of the selfish behaviors at the medium access level on the system performance, with the main result being that such selfish behaviors, even of a small number of users, can paralyze the entire network. Hence a number of mechanisms, based on punishment schemes \cite{Cagalj05}, \cite{Konorski06}, disutility function \cite{Jin07} and Tit-For-Tat concept \cite{chen07a}, have been proposed to thwart the selfish behaviors and drive the network to an efficient equilibrium. Another relevant research thrust \cite{Lee07}, \cite{Chen10}, \cite{Chen09}, consists of applying game theory to analyze and reverse-engineer the 802.11 MAC protocol and providing insights on designing more efficient wireless MAC protocols.


\textit{Non-cooperative rate adaptation:}
Since 802.11 standards do not specify any rate adaptation algorithm, extensive research efforts have been investigated in this field, resulting a number of rate adaptation mechanisms, ranging from the seminal work on automatic rate fallback (ARF) approach \cite{Kamerman97} to later developed  algorithms to improve and replace ARF, e.g., \cite{Lacage04}, \cite{Choi06}, \cite{Singh07}. The key idea is to track channel quality based on packet losses and adapt data rate accordingly.
More recently in non-cooperative setting, related works \cite{Heusse03}, \cite{Tan04}, \cite{Tan05}, \cite{Altman05} have discussed some undesirable effects of the 802.11 MAC layer distributed coordination function (DCF) on the overall network performance when multiple competing nodes use different data rates. Specifically, \cite{Heusse03} demonstrates the so-called \textit{performance anomaly} in 802.11 WLANs such that when competing nodes transmit at different data rates, the aggregate throughput is dominated by the lowest transmission rate. Tan et al. \cite{Tan04} investigate the \textit{time-based fairness}, in which each node is given an equal amount of channel time, and the \textit{throughput-based fairness}, in which each node achieves equal throughputs and show via experiments that time-based fairness can improve performance in multi-rate WLANs. Rate adaptation games in WLANs without and with loss distinction are studied in \cite{Radunovic10}. Pricing mechanisms are proposed in \cite{chen07b} to increase the efficiency of the equilibrium of the non-cooperative rate control game.

\subsection{Paper Overview and Contributions}

Despite a rich body of existing work, the vast majority of them addresses the non-cooperative behaviors at medium access and rate adaptation levels separately. The following natural but crucial questions arise: What is the situation if users can jointly configure their medium access and rate level strategies selfishly? How the two levels interact with each other? How to orient the system towards a fair and efficient equilibrium in this two-dimensional competitive scenario?

Motivated by the above observation, this paper provides a systematic analysis on the non-cooperative joint medium access and rate adaptation game in 802.11 WLANs under the adapted Tit-For-Tat (TFT) framework, a natural strategy in non-cooperative environments widely applied in many applications such as peer-to-peer networks \cite{Qiu04}. We adapt the TFT strategy at both medium access and rate adaptation level (the authors of \cite{chen07a} study the TFT strategy only at medium access level for homogeneous users where the rate adaptation is not considered). The adapted TFT strategy does not require any coordination or incentive mechanisms which may be expensive or even impossible to implement in distributed environments as 802.11 WLANs and inherently ensures user fairness at both levels.

Aiming to provide in-depth understanding of the interaction between rate adaptation and medium access in competitive setting from an analytical perspective, we analyze the structural properties of the formulated game. Specifically, we show that the game has multiple equilibria, which, after the equilibrium refinement process that we develop, reduce to a unique efficient equilibrium. We further develop a distributed algorithm to achieve this equilibrium and demonstrate that the equilibrium achieves the performance very close to the system optimum in a social perspective.

From a user-centric perspective, our work also provides an analytical framework that can stabilize the network in a distributed fashion around a system equilibrium with fairness (or service differentiation) and high efficiency.

\section{System Model}
\label{sec:model}

We consider a multi-rate 802.11 WLAN of a set ${\cal N}=\{1, \cdots, n\}$ of selfish (but not malicious) users, each of whom tries to maximize his own utility (e.g., throughput) by conducting the following two-dimensional selfish strategies:

\textit{Medium access:} Each user $i$ selfishly modifies the distributed random back-off medium access mechanism imposed by 802.11 DCF so as to maximize its utility. More specifically, we focus on the modification of the key parameter in DCF, the contention window (CW) size, denoted by $W_i$.

\textit{Rate adaptation:} Each user $i$ selfishly configures its data rate $R_i$, given the channel condition and other parameters, so as to maximize its utility. Note that IEEE 802.11 standards support multiple transmission rates and leave users to implement their rate adaptation algorithms, which essentially consists of seeking a tradeoff between the transmission rate and the packet loss due to channel errors.

\section{Markov Model of a Multi-rate 802.11 WLAN with Selfish Users}


\subsection{The Markov Model}

To model the users' selfish behaviors and to capture their impact on the network performance, we develop a Markov model on the exponential back-off process by taking into account the users' selfish behaviors. Our model follows the ideas of the seminal work of Bianchi~\cite{Bianchi00}, and is based on the same assumptions:
\begin{compactitem}
\item \textit{Traffic saturation:} The network is saturated such that all users have packets to transmit at any time\footnote{This assumption is justified in our context as the network is saturated when selfish users try to maximize their share of bandwidth by depriving other users of their share.};
\item \textit{Decoupling of collision probability:} The collision probability is decoupled from the back-off stage;
\end{compactitem}

\begin{figure}[htbp]
\centering
\includegraphics[width=7cm]{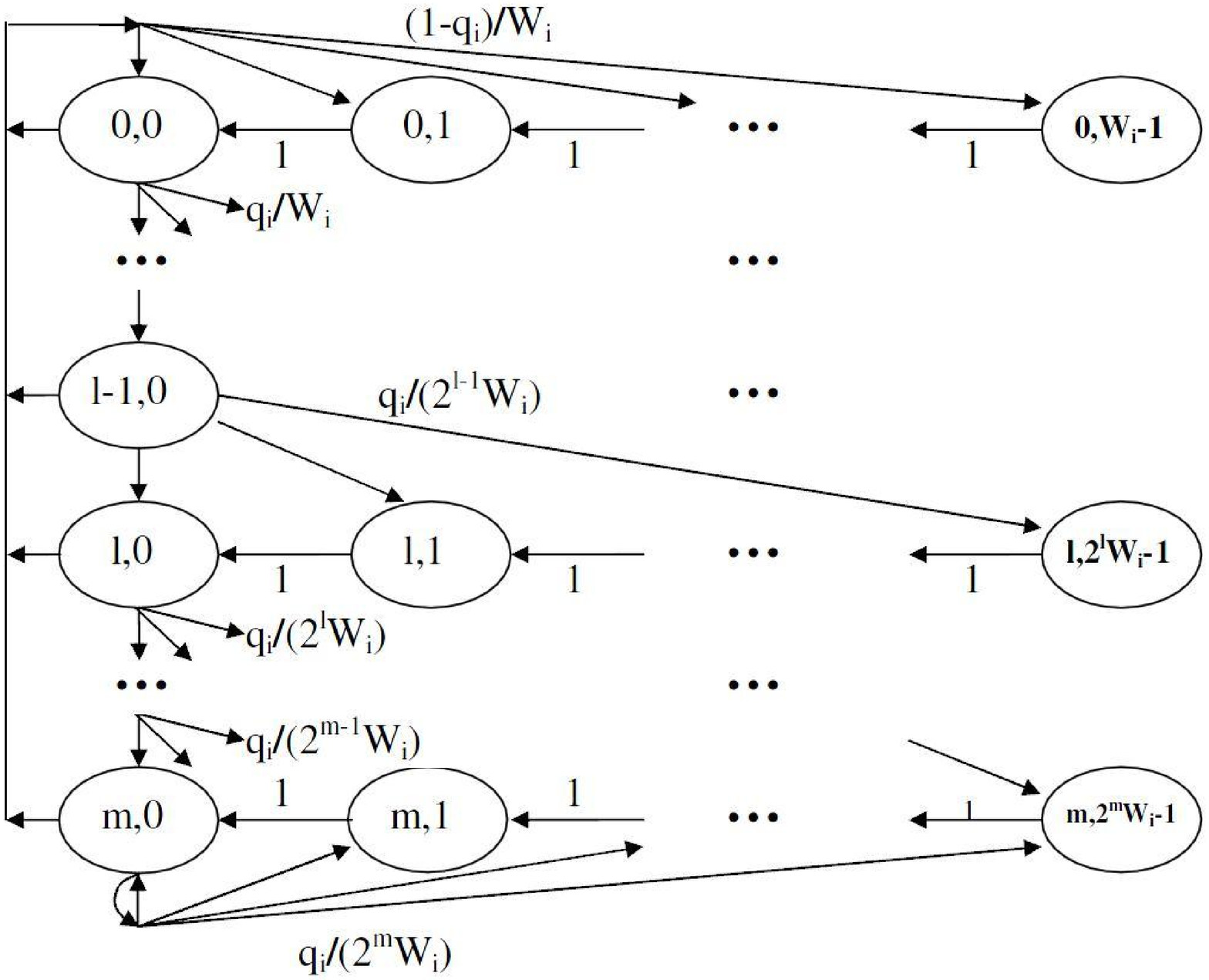}
\caption{The Markov model of user $i$}
\label{fig:markov_chain}
\end{figure}

The Markov model of user $i$ selfishly operating on $W_i$ and $R_i$ is depicted in Figure~\ref{fig:markov_chain}. Each state in the chain is denoted by a couple $(l, k)$ with $l=0,\cdots,m$ being the back-off stage and $k$ the back-off counter ($k\le 2^l-1$). The CW value doubles after each consecutive collision until when the maximum back-off stage $m$ is achieved. State transitions happen at the beginning of each \textit{virtual slot}, as defined in~\cite{Bianchi00}.
Despite the similarity between our model and the Bianchi's model on standard 802.11 DCF, there are
two key differences that make the sequential analysis more challenging:

\textit{Heterogeneity:} The selfishness at both medium access and rate adaptation level leads to heterogeneous Markov models among users. Hence the resulting stationary state solution is heterogeneous among users, rather than being identical for all users as in~\cite{Bianchi00};

\textit{Impact of data rate:} Different from~\cite{Bianchi00}, the stationary state of the Markov model is jointly determined by the medium access parameters and data rate. As shown later in the paper, this cross-layer interdependence brings non-trivial difficulties in the analysis of the system equilibrium and requires an original study that cannot rely on any existing results.

Denote $\tau_i$ the transmission probability of node $i$ in a random slot and $p_i$ the conditioned collision probability of $i$ (the collision probability when $i$ transmits a packet in a random slot), recall the decoupling and saturation assumptions, we can establish the state transition probabilities as follows:
\begin{IEEEeqnarray*}{lll}
P\{l, k | l, k+1\} = 1, &\  0\le k\le 2^lW_i-2, &\ 0\le l\le m, \\
P\{0, k | l, 0\} = \frac{1-q_i}{W_i}, &\ 0\le k\le W_i-1, &\ 0\le l\le m, \\
P\{l, k | l-1, 0\} = \frac{q_i}{2^lW_i}, &\ 0\le k\le 2^lW_i-1, & 1\le l\le m, \\
P\{m, k | m, 0\} = \frac{q_i}{2^mW_i}, &\ 0\le k\le 2^mW_i-1,&
\end{IEEEeqnarray*}
where $P\{s_1,b_1|s_2,b_2\}$ denotes the transition probability from state $(s_1, b_1)$ to $(s_2, b_2)$, $q_i$ denotes the probability a packet arrives at destination without error, detailed as follows
\begin{eqnarray}
q_i=1-(1-p_i)[1-e_i(R_i)],
\label{eq:q}
\end{eqnarray}
where $e_i(R_i)$ is the packet error rate such that a packet is lost due to channel error when user $i$ operates on data rate $R_i$. The four formulas describe the transition probabilities corresponding to the following scenarios under DCF:
\begin{compactitem}
\item $i$ decrements its CW at the beginning of each slot;
\item Once $i$ finish a successful transmission, the next packet's back-off timer is selected uniformly from the range $[0,W_i-1]$, corresponding to the first back-off stage;
\item In case of a transmission failure, the back-off stage is incremented and the back-off timer is selected uniformly from the range $[0,2^lW_i-1]$;
\item Once in stage $m$, the back-off stage is no more incremented.
\end{compactitem}

\subsection{Analysis of Stationary State Solution}

To derive the stationary state of $i$, let $b_{l,k}\triangleq \lim_{t\rightarrow\infty} P\{s(t)=l, b(t)=k\}$ denote the stationary probability distribution, we can derive the following equations to express $b_{l,k}$ by $b_{0,0}$:
\begin{IEEEeqnarray*}{lll}
b_{l,0}=q_ib_{0,0},&\ 0\le l\le m-1, &\\
b_{m,0}=\frac{q_i}{1-q_i}b_{0,0},&\ &\\
b_{l,k}=\frac{2^lW_i-k}{2^lW_i}b_{l,0}, &\ 1\le k\le 2^lW_i-1, &\ 0\le l\le m .
\end{IEEEeqnarray*}

Apply $\sum_{l=0}^m \sum_{k=0}^{2^lW_i-1}b_{l,k}=1$, we can solve $b_{0,0}$ as
\begin{eqnarray*}
b_{0,0}={2(1-2q_i)(1-q_i) \over (1-2q_i)(W_i+1)+q_iW_i(1-(2q_i)^m)}.
\end{eqnarray*}

Thus, the probability $\tau_i$ that user $i$ transmits in a random slot can be expressed as
\begin{eqnarray}
\tau_i=\sum_{l=0}^m b_{l,0}=\frac{2}{W_i+1+q_iW_i\sum_{l=0}^{m-1}(2q)^l}.
\label{eq:tau}
\end{eqnarray}

On the other hand, we have
\begin{eqnarray}
p_i=1-\prod_{j\in{\cal N},j\ne i}(1-\tau_j).
\label{eq:p}
\end{eqnarray}

Combine~\eqref{eq:q}, \eqref{eq:tau} and~\eqref{eq:p}, we obtain $3n$ equations with $3n$ unknowns. In Theorem~\ref{theorem:uniqueness_sol_markov} that follows, we prove that given any data rate profile $(R_i)_{i\in{\cal N}}$, the system characterized by the $3n$ equations admits a unique solution under the mild condition $W_i>3, \forall i\in{\cal N}$. These equations can then be solved numerically.

Before presenting Theorem~\ref{theorem:uniqueness_sol_markov}, we define the following function:
\begin{eqnarray}
\Gamma_i(x)\triangleq \frac{2}{W_i+1+xW_i\sum_{l=0}^{m-1}(2x)^l}.
\label{eq:gamma}
\end{eqnarray}
The following properties hold straightforwardly:
\begin{compactitem}
\item $\tau_i=\Gamma_i(q_i)$;
\item From~\eqref{eq:q}, \eqref{eq:tau} and~\eqref{eq:p}, it holds that
    \begin{eqnarray}
    \frac{(1-q_i)(1-\Gamma_i(q_i))}{1-e_i(R_i)}=\prod_{j\in{\cal N}}(1-\tau_j);
    \label{eq:markov_sol_eq}
    \end{eqnarray}
\item By checking the derivative, it can be shown that $(1-x)[1-\Gamma_i(x)]$ is monotonously decreasing in $x$ if $W_i>3$.
\end{compactitem}

\begin{theorem}
Under the condition $W_i>3, \forall i\in{\cal N}$, the Markov model characterizing a 802.11 WLAN with selfish users admits a unique stationary-state solution for any data rate profile $(R_i)_{i\in{\cal N}}$.
\label{theorem:uniqueness_sol_markov}
\end{theorem}

\begin{proof}
We first show that the Markov model has a stationary state solution. By injecting $p_i$ into $q_i$, we obtain
\begin{eqnarray*}
q_i=1-(1-e_i(R_i))\prod_{j\in{\cal N},j\ne i}(1-\tau_j).
\end{eqnarray*}
Recall \eqref{eq:tau} and denote $\tau_{-i}\triangleq \{\tau_j, j\in{\cal N}, j\ne i\}$, $\tau_i$ can be regarded as a function of $\tau_{-i}$. To prove the existence of a stationary state solution in the Markov chain model, it suffices to show that the mapping from $\tau_{-i}$ to $\tau_i$, described as follows, has a fixed point:
\begin{eqnarray*}
\tau_i=T_i(\tau_{-i})\triangleq \frac{2}{W_i+1+q_iW_i\sum_{l=0}^{m-1}(2q_i)^l},
\end{eqnarray*}
where $q_i=1-(1-e_i(R_i))\prod_{j\in{\cal N},j\ne i}(1-\tau_j)$.

Noticing that $0\le \tau_j\le 1$ holds for any $j\in{\cal N}$ and that $0\le e_i(R_i)\le 1$, it holds that
\begin{eqnarray*}
T_i(\tau_{-i})&\le& \frac{2}{\displaystyle W_i+1+e_i(R_i)W_i\sum_{l=0}^{m-1}\left[2e_i(R_i)\right]^l}< \frac{2}{W_i+1}<1, \\
T_i(\tau_{-i})&\ge& \frac{2}{\displaystyle W_i+1+W_i\sum_{l=0}^{m-1}2^l}>0.
\end{eqnarray*}
It follows from Brouwer fixed point theorem~\cite{Bertsekas07} that there exists a fixed point to the mapping $(T_i)_{i\in{\cal N}}$.

We then proceed to show the uniqueness of the stationary state solution. Assume, by contradiction, that there exists two distinct stationary points $\mathbf{s^1}\triangleq (\tau_i^1,p_i^1,q_i^1, i\in{\cal N})$ and $\mathbf{s^1}\triangleq (\tau_i^2,p_i^2,q_i^2, i\in{\cal N})$, their must exists $i$ such that $q_i^1\ne q_i^2$. Without loss of generality, assume that $q_i^1<q_i^2$.

Note that~\eqref{eq:markov_sol_eq} holds at both solutions and that $(1-x)[1-\Gamma_i(x)]$ is monotonously decreasing in $x$ if $W_i>3$, it follows from the assumption $q_i^1<q_i^2$ that $\prod_{j\in{\cal N}}(1-\tau_j^1)>\prod_{j\in{\cal N}}(1-\tau_j^2)$, which, by applying~\eqref{eq:markov_sol_eq} to other users $j$, leads to $q_j^1<q_j^2$. Noticing that $\tau_j=\Gamma_j(q_j)$ is monotonously decreasing in $q_j$, we have $\tau_j^1>\tau_j^2$ for any other user $j$, which, combined with the assumption $q_i^1<q_i^2$ (thus $\tau_i^1>\tau_i^2$), clearly contradicts with $\prod_{j\in{\cal N}}(1-\tau_j^1)>\prod_{j\in{\cal N}}(1-\tau_j^2)$. We thus complete the proof of the uniqueness of the stationary state solution of the Markov model and also the theorem.
\end{proof}

\section{Adapted Tit-For-Tat Strategy}

First introduced by Anatol Rapoport in Robert Axelrod's two tournaments, Tit-For-Tat (TFT) strategy~\cite{Axelrod84} is regarded as one of the best strategies in non-cooperative environments and is the root of an ever
growing amount of other successful strategies. The core idea of TFT is to start with cooperation and continue to cooperate if the opponents cooperate. The philosophy behind is that in selfish environment each rational player is expected to take more aggressive actions if and only if any other player acts more aggressively.

In our context, we propose the following adapted version of the TFT strategy at two levels.

At the medium access level, the TFT strategy is adapted to ensure that all users operate on the same CW values in order to guarantee the medium access fairness among users. More specifically, each user $i$ measures the CW values of other users during a period of time\footnote{How to observe average CW values during a period of time in a saturated WLAN is addressed in~\cite{Kyasanur05}.}. If it detects another user $j$ operating on a smaller CW value (i.e., $W_j<W_i$), then it reacts by setting $W_i=\min_{j\in{\cal N}} W_j$, otherwise it keeps operating on the previous CW value.

At the rate adaptation level, the TFT strategy is adapted to ensure the fairness among users in terms of channel occupation time, i.e., all users occupy the channel for the same amount of time when transmitting a packet. More specifically, each user $i$ measures the channel occupation time of other users during a period of time. If it detects another user $j$ occupying a longer period of time for transmission (i.e., $T_j>T_i$), then it reacts by setting $T_i=\max_{j\in{\cal N}} T_j$, otherwise it sticks to the previous strategy. The rate adaptation level TFT strategy is essentially motivated by the well-known performance anormaly in 802.11 WLANs where a selfish user tends to transmit at lower data rate so as to enjoy better transmission quality while it penalizes other users since a lower data rate increases its channel occupation time.

The above adapted TFT strategy has following desirable properties which makes it especially suited in WLANs:
\begin{compactitem}
\item The decision is made solely on local measurement.
\item It is simple to implement and only the last measurement needs to be stored.
\item It is especially adapted in wireless networks in that the broadcast nature makes the observation feasible in promiscuous mode.
\item It ensures the fairness among selfish users.
\end{compactitem}


In practice, taking into account the various factors influencing the measurement in wireless environment, a more tolerant version of the TFT strategy called Generous TFT (GTFT) can be applied by integrating a tolerance marge.

\section{Joint medium access and Rate Adaptation Game Formulation}

We focus on the competitive scenario that each user selfishly attempts to optimize its performance by jointly selecting appropriate strategy (CW value and data rate) under the adapted TFT framework depicted in the previously section. This setting gives rise to a non-cooperative joint medium access and rate adaptation game. In this section, we first specify the utility function of each user and then give the formal definition of the game.

\subsection{User Utility: Throughput Analysis}

We study a natural utility function for selfish users, the \textit{effective throughput}, defined as the quantity of bits per unit time successfully arriving at the destination. In this subsection, we use the stationary state solution of the developed Markov model to derive the expected effective throughput of each user $i$. We set out by computing the average virtual slot duration, denoted as $T_{slot}$. As defined in [21], a virtual slot may correspond to a slot where the channel is idle, to a successful transmission, or to a collision. Specifically, let $\sigma$ denote the duration of an empty slot, $T_s^i$ denote the duration of a successful transmission of user $i$, and $T_u^i$ the duration of an unsuccessful transmission of user $i$ due to either collision or channel error. We can mathematically develop $T_{slot}$ as
\begin{multline}
T_{slot} = \prod_{j\in{\cal N}}(1-\tau_j)\sigma + \sum_{i\in{\cal N}}\tau_i\prod_{j\in{\cal N},j\ne i}(1-\tau_j)(1-e_i(R_i))T_s^i
+ \\ \left[\sum_{i\in{\cal N}}\tau_i\!\!\!\prod_{j\in{\cal N},j\ne i}\!\!\!\!\!(1-\tau_j)+\sum_{i\in{\cal N}}\tau_i\left(1-\!\!\!\prod_{j\in{\cal N},j\ne i}\!\!\!\!\!(1-\tau_j)\right)\right]T_u^i,
\label{eq:effective_th}
\end{multline}
where the three terms represent, respectively, the following possible scenarios:
\begin{compactitem}
\item with probability $\prod_{j\in{\cal N}}(1-\tau_j)$, the system experiences an empty slot;
\item with probability $\tau_i\prod_{j\in{\cal N},j\ne i}(1-\tau_j)(1-e_i(R_i))$, $i$ has a successful transmission;
\item with probability $\tau_i\prod_{j\in{\cal N},j\ne i}(1-\tau_j)$, $i$ experiences a collision, with probability $\tau_i\prod_{j\in{\cal N},j\ne i}(1-\tau_j)e_i(R_i)$, $i$ experiences a transmission failure due to channel error.
\end{compactitem}

In the base-line 802.11 MAC protocol without RTS/CTS\footnote{Although we focus in this work on the base-line model, our method and analysis can be applied in the model with RTS/CTS dialogue, which we leave for future study.}, the TFT framework ensures that the packet duration is identical among users. Denote $T$ the time for a packet transmission, denote ACK, DIFS and SIFS the time to transmit the ACK DIFS and SIFS frames, respectively. By neglecting the propagation delay, we can compute $T_s^i$ and $T_u^i$ as follows
\begin{eqnarray*}
T_s^i &=& T+SIFS+ACK+DIFS\simeq T, \\
T_u^i &=& T+SIFS \simeq T,
\end{eqnarray*}
where the approximation is due to that the control frame size is order of magnitude smaller than that of a data frame.

Let $\rho_i\triangleq 1-\tau_i$, noting that $\sigma\ll T$, we have
\begin{eqnarray*}
T_{slot}\simeq (1-\prod_{j\in{\cal N}}\rho_j)T.
\end{eqnarray*}

Note that (1) the probability that user $i$ performs a transmission without collision is $(1-\rho_i)\prod_{j\in{\cal N}, j\ne i}\rho_j$; (2) the transmitted packet, whose size is $R_iT$, is corrupted due to channel error with probability $e_i(R_i)$, after some algebraic operations, the effective throughput can be written as
\begin{eqnarray*}
S_i=\frac{(1-\rho_i)\prod_{j\in{\cal N}, j\ne i}\rho_j(1-e_i(R_i))R_i}{1-\prod_{j\in{\cal N}}\rho_j}.
\end{eqnarray*}

To conclude this subsection, we take an analytical look at the packet error rate function $e_i(R_i)$. Intuitively, under the adapted TFT framework, each user should strike a balance between sending more bits in the packet transmission time by choosing a higher data rate but at the price of increasing the transmission error and operating at a lower data rate with less transmission error. Mathematically, such tradeoff at the rate adaptation level is modeled as follows. Assuming perfect error detection and no error correction, we can express $e_i(r_i)$ as $e_i(R_i)=(1-P_e(R_i))^{R_iT}$, where $P_e(R_i)$ is the bit error rate (BER). $P_e$ is a function of $E_b/N_0$, the bit-energy-to-noise ratio of the received signal, e.g., $P_e={1\over 2}e^{-E_b/N_0}$ for DPSK. Assuming an additive white Gaussian noise (AWGN) channel, the bit-energy-to-noise ratio of $i$ $\left({E_b \over N_0}\right)_i$ of the received signal is derived from the SNR (Signal-to-Noise Ratio) as follows:
$$\displaystyle \left({E_b \over N_0}\right)_i = SNR {B_t \over C_i} = {h_iP_i \over \sigma^2} {B_t \over R_i}
= {h_iB_t \over \sigma^2} {P_i \over C_i},$$
where $B_t$ is the unspread bandwidth of the signal, $P_i$ is the transmission power of $i$, $h_i$ is the channel gain from $i$ to the receiver and
$\sigma^2$ is the AWGN power at the receiver.

Generally, define $G_i(R_i)\triangleq (1-e_i(R_i))R_i$, we have the following features on $e_i(R_i)$ and $G_i(R_i)$ with practical setting on the packet transmission duration $T$ and the minimal and maximal operational data rate $R_{min,i}$ and $R_{max,i}$.

\begin{lemma}
It holds that:
\begin{compactitem}
\item $e_i(0)\rightarrow 0$, $e_i(\infty)\rightarrow 1$; $e_i(R_i)$ is monotonously increasing, twice derivable and strictly convex in $R_{min,i}\le R_i\le R_{max,i}$;
\item $e_i'(R_i)>0$ and is monotonously increasing in $R_i$;
\item $G_i(0)=0$, $G_i(\infty)\rightarrow 0$; $G_i(R_i)$ is strictly concave in $R_{min,i}\le R_i\le R_{max,i}$.
\end{compactitem}
\label{lemma:channel_error_rate}
\end{lemma}

\subsection{Game Formulation}

We model the interactions among selfish users as a non-cooperative joint medium access and rate adaptation game, in which each user selfishly maximizes its effective throughput by jointly selecting its CW value and data rate under the adapted TFT framework. The game is formally defined as follows\footnote{Throughout the paper, we use \textit{player} and \textit{user} interchangeably to denote a wireless node in the studied WLAN.}.

\begin{definition}
The non-cooperative joint medium access and rate adaptation game, denoted as $G$, is defined as a triple $G\triangleq ({\cal N}, ({\cal W}_i,{\cal R}_i)_{i\in{\cal N}}, (S_i)_{i\in{\cal N}})$, where ${\cal N}$ is the set of players, $({\cal W}_i,{\cal R}_i)$ is the strategy set of player $i$, where ${\cal W}_i\triangleq \{0,\cdots,\infty\}$ and ${\cal R}_i\triangleq [R_{min,i},R_{max,i}]$, the effective throughput $S_i$ is the utility function of player $i$.
\end{definition}

The solution of the game is characterized by a Nash equilibrium~\cite{Myerson91} (NE, simply noted as equilibrium in the paper), a strategy profile from which no player has incentive to deviate unilaterally.
It can be noted that the challenge in solving $G$ is to solve a two-dimensional and non-decomposable optimization problem for each player. To overcome this difficulty and to get more insight on the structure of resulting equilibrium, we introduce the following two-level hierarchical game model: the lower-level rate adaptation game under fixed CW setting and the higher-level medium access game.

\begin{definition}
Given a fixed CW setting $\mathbf{W}\triangleq (W_i)_{i\in{\cal N}}$ at the medium access level, the non-cooperative rate adaptation game, denoted as $G_R(\mathbf{W})$, is a tuple $({\cal N}, ({\cal R}_i)_{i\in{\cal N}}, (S_i)_{i\in{\cal N}})$, where $\cal N$ the player set, ${\cal R}_i=[R_{min,i},R_{max,i}]$ is the strategy set of player $i$, $S_i$ is the utility function of $i$. Each player $i$
selects its data rate $R_{min,i}\le R_i\le R_{max,i}$ to maximize $S_i$.
\label{def:rate_game}
\end{definition}

\begin{definition}
The non-cooperative medium access game, denoted as $G_M$, is a tuple $({\cal N}, ({\cal W}_i)_{i\in{\cal N}}, (\widehat{S}_i)_{i\in{\cal N}})$, where $\cal N$ the player set, ${\cal W}_i= \{0, \cdots, \infty\}$ is the strategy set of player $i$, $\widehat{S}_i$ is the utility function of $i$, defined as $\widehat{S}_i(W_i, W_{-i})\triangleq S_i(\mathbf{W}, \mathbf{R^*(W)})$, where $\mathbf{R^*(W)}$ denotes the equilibrium of $G_R(\mathbf{W})$ (i.e., the data rate profile at the equilibrium of the rate adaptation game under $\mathbf{W}$). Each player $i$ selects its strategy $W_i\in{\cal W}_i$ to maximize its utility function $\widehat{S}_i$.
\label{def:medium access_game}
\end{definition}

By decomposing $G$, we introduce a two-level hierarchical
architecture which will help us analyze the two-dimensional joint medium access and rate adaptation problem in the non-cooperative setting, as explored in the next two sections.

\section{Equilibrium Analysis of Lower-level Rate Adaptation Game}

In this section, we solve the lower-level non-cooperative rate adaptation game in which each user adapts its data rate to maximize its effective throughput under fixed CW setting. We study the existence and uniqueness of equilibrium and explore some structural properties of the equilibrium. We focus our attention to scenarios where users play pure strategies (do
not randomized their data rate). We obtain the following results: (1) a pure NE exists and is unique; (2) we establish the convergence to the unique NE; (3) we investigate the efficiency of the NE in terms of Price of Anarchy (PoA) to gain more insights into the system behaviors.

To make the analysis tractable, we make the following assumption (or approximation): the number of users is large enough so that the strategy change of one user has neglectable influence on the system state. The approximation is accurate in large systems where the micro-level strategy variation of any individual player has neglectable influence on the macro-level system state. From a game theoretic perspective, we consider a \textit{non-atomic game}\footnote{Please refer to~\cite{Schmeidler73} for a detailed presentation on non-atomic games.}.

Particularly in this section, the assumption is mathematically expressed as the following approximation:
\begin{eqnarray}
\prod_{j\in{\cal N}}\rho_j\simeq \prod_{j\in{\cal N},j\ne i}\rho_j.
\label{eq:atomic}
\end{eqnarray}

Lemma~\ref{lemma:property} follows immediately from the non-atomic game assumption. The proof, consists of checking the relevant derivatives and applying Lemma~\ref{lemma:channel_error_rate} and the equations~\eqref{eq:q}, \eqref{eq:tau}, \eqref{eq:p} and $\rho_i=1-\tau_i$.

\begin{lemma}
Under non-atomic game assumption, we have
\begin{compactitem}
\item $\rho_i$ is monotonously decreasing and convex in $R_i$;
\item $\frac{\partial \tau_i}{\partial q_i}$ is negative and monotonously increasing in $R_i$.
\end{compactitem}
\label{lemma:property}
\end{lemma}

\begin{proof}
It follows from the non-atomic game assumption that $p_i=\prod_{j\in{\cal N},j\ne i} (1-\tau_j)$ is independent to $R_i$. As $e_i(R_i)$ is monotonously increasing in $R_i$ (Lemma~\ref{lemma:channel_error_rate}), it follows from~\eqref{eq:q} that $q_i$ is monotonously increasing in $R_i$. It then follows from~\eqref{eq:tau} that $\tau_i$ is monotonously increasing in $R_i$. Hence $\rho_i=1-\tau_i$ is monotonously decreasing in $R_i$.
\end{proof}

\subsection{Best-response Function}

At an equilibrium of $G_R$, if there exists, the data rate of each player $i$ consists of the \textit{best response} given the strategy setting of others. Formally, the best-response correspondence of player $i$ expressed as a function of the strategies of other players $R_{-i}$, denoted as $B_i(R_{-i})$, is defined as follows
\begin{eqnarray}
B_i(R_{-i})\triangleq \max_{R_{min,i}\le R_i\le R_{max,i}} S_i(R_i,R_{-i}).
\end{eqnarray}
We prove in this subsection that in the rate adaptation game $G_R$, $B_i(R_{-i})$ is uniquely determined by $R_{-i}$.\footnote{Generally speaking, the best response strategy of a game is not necessarily a function (one-to-one mapping), i.e., there may exist several global maxima of the utility function.} Hence $B_i(R_{-i})$ can be defined as a function of $R_{-i}$.

We start by investigating the derivative of the utility function. It can be noted that $S_i(R_i, R_{-i})$ is continuous and differentiable in $R_i$. After some straightforward mathematical development, we can write $\frac{\partial S_i}{\partial R_i}$ as follows:
\begin{multline}
\frac{\partial S_i}{\partial R_i} = \frac{(1-\rho_i)G_i(R_i)\prod_{j\in{\cal N},j\ne i}\rho_j}{1-\prod_{j\in{\cal N}}\rho_j}\left[\frac{G_i'(R_i)}{G_i(R_i)}-\right. \\
\left.\frac{1-\prod_{j\in{\cal N},j\ne i}\rho_j}{(1-\rho_i)(1-\prod_{j\in{\cal N}}\rho_j)}\frac{\partial \rho_i}{\partial R_i}\right].
\label{eq:derivative1}
\end{multline}
Let $A_i(R_i)$ denote the term in the parenthesis of~\eqref{eq:derivative1}:
\begin{eqnarray}
A(R_i)\triangleq \frac{G_i'(R_i)}{G_i(R_i)}-\frac{1-\prod_{j\in{\cal N},j\ne i}\rho_j}{(1-\rho_i)(1-\prod_{j\in{\cal N}}\rho_j)}\frac{\partial \rho_i}{\partial R_i}.
\label{eq:a}
\end{eqnarray}
Noticing that $\rho_i=1-\tau_i$ and~\eqref{eq:q}, we have:
\begin{eqnarray*}
\frac{\partial \rho_i}{\partial R_i}=-\frac{\partial \tau_i}{\partial q_i}\frac{\partial q_i}{\partial R_i}=-\frac{\partial \tau_i}{\partial q_i}\prod_{j\in{\cal N},j\ne i}\rho_je_i'(R_i).
\end{eqnarray*}
It then follows that
\begin{eqnarray}
A(R_i)= \frac{G_i'(R_i)}{G_i(R_i)}+\frac{\left(1-\prod_{j\in{\cal N},j\ne i}\rho_j\right)\prod_{j\in{\cal N},j\ne i}\rho_j}{(1-\rho_i)(1-\prod_{j\in{\cal N}}\rho_j)}\frac{\partial \tau_i}{\partial q_i}e_i'(R_i).
\label{eq:a2}
\end{eqnarray}

\begin{lemma}
$A_i(R_i)$ is monotonously decreasing in $R_i$.
\label{lemma:monotonicity_a}
\end{lemma}

\begin{proof}
We can develop the second term of $A_i(R_i)$ as:
\begin{multline*}
\frac{1-\prod_{j\in{\cal N},j\ne i}\rho_j}{(1-\rho_i)(1-\prod_{j\in{\cal N}}\rho_j)}\frac{\partial \rho_i}{\partial R_i}= \left(\frac{1}{1-\rho_i}-\frac{\prod_{j\in{\cal N},j\ne i}\rho_j}{1-\rho_i\prod_{j\in{\cal N},j\ne i}\rho_j}\right)\cdot \\ \frac{\partial \rho_i}{\partial R_i}
=\left(\frac{\partial \log \left(\frac{1}{\prod_{j\in{\cal N},j\ne i}\rho_j}-\rho_i\right)}{\partial \rho_i}-\frac{\partial \log (1-\rho_i)}{\partial \rho_i}\right)\frac{\partial \rho_i}{\partial R_i}\\
=\frac{\partial \log \frac{\frac{1}{\prod_{j\in{\cal N},j\ne i}\rho_j}-\rho_i}{1-\rho_i}}{\partial \rho_i}\frac{\partial \rho_i}{\partial R_i}.
\end{multline*}

It can be checked that since $\displaystyle \!\!\!\!\!\prod_{j\in{\cal N},j\ne i}\!\!\!\!\!\rho_j\le 1$, $\log \frac{\frac{1}{\prod_{j\in{\cal N},j\ne i}\rho_j}-\rho_i}{1-\rho_i}$ is monotonously increasing and convex in $0\le \rho_i\le 1$. Since $\rho_i$ is convex and monotonously decreasing in $R_i$ (Lemma~\ref{lemma:property}), $\frac{\partial \log \frac{\frac{1}{\prod_{j\in{\cal N},j\ne i}\rho_j}-\rho_i}{1-\rho_i}}{\partial \rho_i}\frac{\partial \rho_i}{\partial R_i}$ is thus monotonously increasing in $R_i$. On the other hand, it follows from Lemma~\ref{lemma:channel_error_rate} that $\frac{G_i'(R_i)}{G_i(R_i)}=(\log G_i(R_i))'$ is monotonously decreasing in $R_i$. Hence $A_i(R_i)$ is monotonously decreasing in $R_i$.
\end{proof}

Based on Lemma~\ref{lemma:monotonicity_a}, we can drive the maximizer of $S_i$ by distinguishing the following three cases:
\begin{compactenum}
\item $A_i(R_{min,i})\le 0$. It follows from~\eqref{eq:derivative1} that $\frac{\partial S_i}{\partial R_i}<0$ for $R_{min,i}<R_i\le R_{max,i}$. $S_i(R_i)$ is monotonously decreasing in $[R_{min,i}, R_{max,i}]$ with a unique maximizer being $R_{min,i}$;
\item $A_i(R_{max,i})\ge 0$. It follows from~\eqref{eq:derivative1} that $\frac{\partial S_i}{\partial R_i}>0$ for $R_{min,i}\le R_i< R_{max,i}$. $S_i(R_i)$ is monotonously increasing in $[R_{min,i}, R_{max,i}]$ with a unique maximizer being $R_{max,i}$;
\item $A_i(R_{min,i})>0$ and $A_i(R_{max,i})<0$. There exists $R_i^*$ such that $\frac{\partial S_i}{\partial R_i}>0$ for $R_{min,i}\le R_i\le R^*$ and $\frac{\partial S_i}{\partial R_i}<0$ for $R_i^*\le R_i\le R_{max,i}$. $R_i^*$ is thus the unique maximizer of $S_i$.
\end{compactenum}

The above analysis shows that the best-response in $G_R$ is indeed a function that can be expressed as follows:
\begin{eqnarray}
B_i(R_{-i})=
\begin{cases}
R_{min,i} & \text{if } A_i(R_{min,i})\le 0, \\
A_i^{-1}(0) & \text{if } A_i(R_{min,i})>0, A_i(R_{max,i})<0,\\
R_{max,i} & \text{if } A_i(R_{max,i})\ge 0.
\end{cases}
\label{eq:best_response}
\end{eqnarray}

Defined $\mathbf{B(\mathbf{R})}\triangleq (B_i(R_{-i}))_{i\in{\cal N}}$, we show in Lemma~\ref{lemma:increasing_b} that $\mathbf{B(\mathbf{R})}$ is non-decreasing in $R_j, \forall j\in{\cal N}, j\ne i$.

\begin{lemma}
$\mathbf{B(\mathbf{R})}$ is non-decreasing in $R_j, \forall j\in{\cal N}, j\ne i$.
\label{lemma:increasing_b}
\end{lemma}

\begin{proof}
Recall the best-response function~\eqref{eq:best_response}, let $R^*_i$ denote the root of the equation $A_i(R_i)=0$ where $A_i(R_i)$ is defined in~\eqref{eq:a}, it suffices to show that if $R_{min,i}<R^*_i<R_{max,i}$, then $R^*_i$ is increasing in $R_j, \forall j\in{\cal N}, j\ne i$. We proceed our proof by contradiction. Assume, by contradiction, that $R_i^a$ and $R_i^b$ satisfying $A_i(R_i)=0$ such that $R_i^a<R_i^b$ and $R_j^a\ge R_j^b$, $\forall j\in{\cal N}, j\ne i$.

We first show that under the assumption, it holds that $\rho_i^a<\rho_i^b$. We proceed by distinguishing the following two cases:
\begin{compactitem}
\item \textit{Case 1:} $\prod_{j\in{\cal N}}\rho_j^a\ge \prod_{j\in{\cal N}}\rho_j^b$. It follows from eq 18 and the assumption $R_i^a<R_i^b$ that $q_i^a<q_i^b$. It then follows from~~\eqref{eq:tau} that $\tau_i^a>\tau_i^b$. Hence $\rho_i^a<\rho_i^b$.
\item \textit{Case 2:} $\prod_{j\in{\cal N}}\rho_j^a< \prod_{j\in{\cal N}}\rho_j^b$. It follows from eq 18 and the assumption $R_j^a\ge R_j^b, \forall j\in{\cal N}, j\ne i$ that $q_j^a>q_j^b$. It then follows from~~\eqref{eq:tau} that $\tau_j^a<\tau_j^b$. Hence $\rho_j^a>\rho_j^b, \forall j\in{\cal N}, j\ne i$. It then follows $\prod_{j\in{\cal N}}\rho_j^a< \prod_{j\in{\cal N}}\rho_j^b$ that $\rho_i^a<\rho_i^b$.
\end{compactitem}

We next show that under the assumption, it holds that $\prod_{j\in{\cal N},j\ne i}\rho_j^a\ge \prod_{j\in{\cal N},j\ne i}\rho_j^b$. Otherwise, if $\prod_{j\in{\cal N},j\ne i}\rho_j^a> \prod_{j\in{\cal N},j\ne i}\rho_j^b$, then recall the assumption on the non-atomic game (equation~\eqref{eq:atomic}), it follows from~\eqref{eq:markov_sol_eq} that
\begin{eqnarray*}
\frac{(1-q_j^a)(1-\Gamma_j(q_j^a))}{1-e_j(R_j^a)}=\prod_{l\in{\cal N},l\ne i}\rho_j^a
>\frac{(1-q_j^b)(1-\Gamma_j(q_j^b))}{1-e_j(R_j^b)}\\
=\prod_{l\in{\cal N},l\ne i}\rho_j^b, \forall j\in{\cal N}, j\ne i.
\end{eqnarray*}
Since $e_j(R_j)$ is increasing in $R_j$ and we have shown that $(1-q_j)(1-\Gamma_j(q_j))$ is decreasing in $q_j$, it holds that $q_j^a<q_j^b, \forall j\in{\cal N}, j\ne i$, leading to $\rho_j^a>\rho_j^b, \forall j\in{\cal N}, j\ne i$, which contradicts to the assumption that $\prod_{j\in{\cal N},j\ne i}\rho_j^a> \prod_{j\in{\cal N},j\ne i}\rho_j^b$.

Until now we have shown that $\rho_i^a<\rho_i^b$ and $\prod_{j\in{\cal N},j\ne i}\rho_j^a\ge \prod_{j\in{\cal N},j\ne i}\rho_j^b$. Now recall the atomic assumption, we can derive from~\eqref{eq:a2} that
\begin{eqnarray*}
A_l(R_l)= \frac{g_l'(R_l)}{g_l(R_l)}+\frac{\prod_{j\in{\cal N},j\ne l}\rho_j}{1-\rho_i}\frac{\partial \tau_l}{\partial q_l}e_l'(R_l).
\end{eqnarray*}
Noticing the assumption $R_i^a<R_i^b$, it holds that $A_i(R_i^a,R_{-i}^a)>A_i(R_i^b,R_{-i}^b)$, which contradicts with the fact both $R_i^a$ and $R_i^b$ satisfy $A_i(R_i)=0$. This contradiction shows that the best-response mapping $B_i(R_{-i})$ is non-decreasing in $R_j, \forall j\in{\cal N}, j\ne i$.
\end{proof}

\subsection{Equilibrium Analysis: Existence, Uniqueness and Convergence}

Armed with the analysis on $B_i(R_{-i})$, we now show that $G_R$ has a unique equilibrium.

\begin{theorem}
Given any CW setting $\mathbf{W}$, the non-cooperative rate adaptation game $G_R(\mathbf{W})$ admits at least an equilibrium.
\label{theorem:ne_existence_rate_game}
\end{theorem}

\begin{proof}
Recall Lemma~\ref{lemma:increasing_b} and that $B_i(R_{-i})$ is bounded such that $R_{min,i}\le B_i(R_{-i})\le R_{max,i}$, starting by $(R_i=R_{min,i})_{i\in{\cal N}}$, the best response mapping must converge in a monotonously non-decreasing fashion to a fixed point which is an equilibrium of the game.
\end{proof}

Theorem~\ref{theorem:ne_uniqueness_rate_game} further establishes the uniqueness of NE.

\begin{theorem}
Given any CW setting $\mathbf{W}$, the non-cooperative rate adaptation game $G_R(\mathbf{W})$ admits a unique equilibrium.
\label{theorem:ne_uniqueness_rate_game}
\end{theorem}

\begin{proof}
We prove the theorem by contradiction. Assume that two distinct equilibria $\mathbf{a}$ and $\mathbf{b}$ exist. Without loss of generality, assume that $\prod_{j\in{\cal N}} \rho_j^a\le \prod_{j\in{\cal N}} \rho_j^b$.

We first show, by contradiction, that $\rho_j^a\le \rho_j^b, \forall j\in{\cal N}$. Assume that there exists $i\in{\cal N}$ such that $\rho_i^a>\rho_i^b$. We show that $R_i^a\le R_i^b$. To this end, assume, by contradiction, that $R_i^a>R_i^b$. We have shown in the proofs of the previous two theorems that $A_i(R_i)$ is monotonously decreasing in $R_i$ and $\rho_i$, monotonously increasing in $\rho_j$ ($j\in {\cal N}, j\ne i$). It then follows from the assumption that
\begin{eqnarray*}
A_i(R_i^a)<A_i(R_i^b).
\end{eqnarray*}
Now recall that $R_i^a$ and $R_i^b$ is the maximizer of $S_i$ is the range $[R_{min,i},R_{max,i}]$, noticing the relationship between $\frac{\partial S_i}{\partial R_i}$, we have
\begin{eqnarray*}
R_i^a>R_i^b \Rightarrow R_i^a>R_{min,i} \Rightarrow A_i(R_i^a)\ge 0.
\end{eqnarray*}
We thus have $A_i(R_i^b)>0$, hence we must have $R_i^b=R_{max,i}$, which contradicts with the assumption $R_i^a>R_i^b$ since $R_a^i$ is upper-bounded by $R_{max,i}$. Hence we have $R_i^a \le R_i^b$. Recall the assumption that $\prod_{j\in{\cal N}} \rho_j^a\le \prod_{j\in{\cal N}} \rho_j^b$ and $\rho_i^a>\rho_i^b$, it follows from~\eqref{eq:q}, \eqref{eq:p}, \eqref{eq:tau} and $\rho_i=1-\tau_i$ that $\rho_i^a<\rho_i^b$, which contradicts the assumption $\rho_i^a>\rho_i^b$. We have thus shown that $\rho_j^a\le \rho_j^b, \forall j\in{\cal N}$.

It then holds that $\tau_j^a\ge \tau_j^b, \forall j\in{\cal N}$, it then follows from~\eqref{eq:markov_sol_eq} that $e_j(R_j^a)\le e_j(R_j^b)$, leading to $R_j^a\le R_j^b$. If $R_j^a= R_j^b, \forall j\in{\cal N}$, then recall Theorem~\ref{theorem:uniqueness_sol_markov}, it holds that the two equilibria $\mathbf{a}$ and $\mathbf{b}$ are identical, which contradicts to the assumption that they are distinct equilibria. Hence there must exist $l\in{\cal N}$ such that $R_l^a<R_l^b$. We distinguish the following two cases
\begin{compactitem}
\item \textit{Case 1:} $R_l^b=R_{max,l}$. It holds that $A_i(R_l^b)\ge 0$. If $R_l^a=R_{min,l}$, then $A_l(R_l^a)\le 0$. If $R_l^a>R_{min,l}$, then $A_l(R_l^a)=0$;
\item \textit{Case 1:} $R_l^b<R_{max,l}$. Since $R_l^b>R_l^a\ge R_{min,i}$, it holds that $A_l(R_l^b)=0$. If $R_l^a=R_{min,l}$, then $A_l(R_l^a)\le 0$. If $R_l^a>R_{min,l}$, then $A_l(R_l^a)=0$.
\end{compactitem}
In both cases, we have $A_l(R_l^b)\ge A_l(R_l^a)$.

On the other hand, recall the atomic assumption, we can derive from~\eqref{eq:a2} that
\begin{eqnarray*}
A_l(R_l)= \frac{g_l'(R_l)}{g_l(R_l)}+\frac{\prod_{j\in{\cal N},j\ne l}\rho_j}{1-\rho_i}\frac{\partial \tau_l}{\partial q_l}e_l'(R_l).
\end{eqnarray*}
For each $i\in{\cal N}$, we have shown that $R_i^a\le R_i^b$ (with $R_l^a<R_l^b$), $\rho_i^a\le \rho_i^b$, $\forall i\in{\cal N}$. Hence $\frac{\prod_{j\in{\cal N},j\ne l}\rho_j^a}{1-\rho_i^a}\le \frac{\prod_{j\in{\cal N},j\ne l}\rho_j^b}{1-\rho_i^b}$. It also follows from $\rho_i=1-\tau_i$ and~\eqref{eq:tau} that $q_i^a\le q_i^b$. Noticing that $e_l'(R_l)>0$ and is monotonously increasing in $R_l$ (Lemma~\ref{lemma:channel_error_rate}) and that $\tau_l$ is concave and monotonously decreasing in $q_l$ meaning that $\frac{\partial \tau_l}{\partial q_l}>0$ and is monotonously decreasing in $q_l$, it holds that
$$\left.\frac{\prod_{j\in{\cal N},j\ne l}\rho_j}{1-\rho_i}\frac{\partial \tau_l}{\partial q_l}e_l'(R_l)\right|^a>\left.\frac{\prod_{j\in{\cal N},j\ne l}\rho_j}{1-\rho_i}\frac{\partial \tau_l}{\partial q_l}e_l'(R_l)\right|^b$$
Moreover, since $g_l(R_l)$ is concave in $R_l$, $\log g_l(R_l)$ is concave in $R_l$. Hence $\frac{g_l'(R_l)}{g_l(R_l)}=\frac{\partial\log g_l(R_l)}{\partial R_l}$ is decreasing in $R_l$. Therefore, we have $A_l(R_l^b)<A_l(R_l^a)$, which contradicts with $A_l(R_l^b)\ge A_l(R_l^a)$. This contradiction completes our proof on the equilibrium uniqueness.
\end{proof}

The following theorem follows naturally, establishing the convergence to the unique equilibrium of $G_R(\mathbf{W})$ under any initial state $(R_i(0))_{i\in{\cal N}}$ in an \textit{asynchronous} manner.

\begin{theorem}
Assume that the players of $G_R(\mathbf{W})$ follows the best response dynamics from some initial rate vector $\mathbf{R}(0)\triangleq (R_i(0))_{i\in{\cal N}}$,
i.e., from time to time, an asynchronous update step is taken where some user $i\in{\cal N}$ updates its strategy from $R_i(n)$ to $R_i(n+1)=B_i(R_i(n))$, and the sequence of players doing the
update steps can be arbitrary as long as the number of steps between consequent updates of every individual player is bounded. Then, $\lim_{n\rightarrow\infty} \mathbf{R}(n)=\mathbf{R}^*$, where $\mathbf{R}^*$ is the unique NE of $G_R(\mathbf{W})$.
\end{theorem}

\begin{proof}
First, consider an arbitrary sequence of update steps commencing from an initial vector $\mathbf{R}(0)=(R_{min,i}, i\in{\cal N})$, and denote the resulting sequence of rate vectors by
$\mathbf{R}_{min}(n)$. Obviously, for any player $i$, the first time it updates its strategy will be a nondecreasing update. In light of Lemma~\ref{lemma:increasing_b}, it follows by induction that all updates
must be nondecreasing, i.e., $\mathbf{R}_{min}(n)$ is a nondecreasing sequence. Since $\mathbf{R}_{min}(n)\le \mathbf{R}_{max}$ is bounded, it must converge to a limit. Due to the
continuity of the best-response function, this limit must be its (unique) fixed point $\mathbf{R}^*$.

In a similar manner, consider a sequence of best-response updates from an initial vector of $\mathbf{R}(0)=(R_{max,i}, i\in{\cal N})$. By the same analysis, Lemma~\ref{lemma:increasing_b} implies that all the updates in
the sequence must be non-increasing, and the sequence must therefore converge to $\mathbf{R}^*$.

Finally, consider a sequence of best-response updates $\mathbf{R}(n)$ commencing from an arbitrary initial data rate vector $\mathbf{R}(0)$. From Lemma~\ref{lemma:increasing_b}, it follows that $\mathbf{R}_{min}(n)\le \mathbf{R}(n)\le \mathbf{R}_{max}(n)$, provided that for every $n$, the update step is performed by the same flow in all three sequences. Since, as established above, $\mathbf{R}_{min}(n)$ and $\mathbf{R}_{max}(n)$ converge to $\mathbf{R}^*$, it follows that the same is true for $\mathbf{R}(n)$ as well.
\end{proof}

\subsection{Equilibrium Efficiency Analysis}

Having derived the unique NE of $G_R(\mathbf{W})$, we proceed to evaluate the efficiency of the unique NE compared with the global optimal by quantifying the \textit{Price of Anarchy} (PoA), [21], defined as the ratio between the optimal global utility and the system utility achieved at the NE. In other words, PoA quantifies the efficiency loss due to selfish competition compared to cooperation. Due to the complexity of the problem and in order to make our analysis tractable, we focus on the symmetrical and unconstraint scenario where $G_i(R_i)$ are identical for all players and $R_{min}$ ($R_{max}$, respectively) are sufficiently small (large) so that the unique equilibrium satisfies $\frac{\partial S_i}{\partial R_i}=0$. However, as shown in Section~\ref{sec:simu} by simulation, we observe the same results for general cases.

We first show in Theorem~\ref{theorem:ne_sym_rate} that in the symmetrical case, the equilibrium is also symmetrical such that it holds that $R_i^*=R_j^*, \rho_i^*=\rho_j^*, \forall i,j\in{\cal N}$.

\begin{theorem}
In the symmetrical case, the equilibrium is also symmetrical, i.e., $R_i^*=R_j^*, \rho_i^*=\rho_j^*, \forall i,j\in{\cal N}$.
\label{theorem:ne_sym_rate}
\end{theorem}

\begin{proof}
We first prove $R_i^*=R_j^*, \forall i,j\in{\cal N}$. Otherwise if there exist $i,j$ such that $R_i^*\ne R_j^*$. Since the players are symmetrical, by swapping $R_i^*$ and $R_j^*$ we obtain another equilibrium. This clearly contradicts with Theorem~\ref{theorem:ne_uniqueness_rate_game} on the uniqueness of the equilibrium.

Noticing that $e_i$ are identical among players, it then follows immediately from 18 that $\tau^*_i=\tau_j^*, \forall i,j\in{\cal N}$, leading to $\rho_i^*=\rho_j^*, \forall i,j\in{\cal N}$.
\end{proof}

Armed with Theorem~\ref{theorem:ne_sym_rate}, the equilibrium can thus be solved numerically by imposing $A_i(R_i)=0$ and combining the equations~\eqref{eq:q}, \eqref{eq:tau} and \eqref{eq:p}. Specifically we study a practical setting where users use $M$-ary QAM modulation of which the bit error rate is
\begin{eqnarray}
P_e\simeq 4\left(1-\frac{1}{\sqrt{M}}\right)Q\left(\sqrt{\frac{3\alpha^2\log_2(M)}{(M-1)}\frac{E_b}{N_0}}\right).
\label{eq:qam}
\end{eqnarray}

\begin{figure}[htbp]
\centering
\includegraphics[width=4.5cm, angle=270]{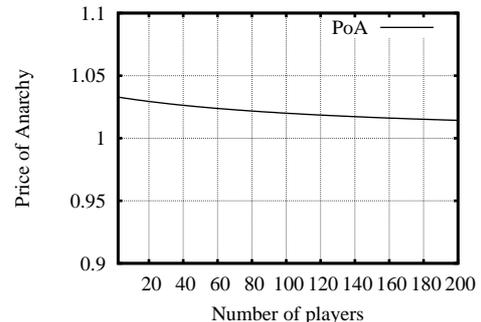}
\vspace{-0.3cm}
\caption{Price of Anarchy as a function of $N$}
\label{fig:poa1}
\end{figure}

In Figure~\ref{fig:poa1}, we plot the numerical result on the PoA with a typical WLAN setting $M=64$, $\alpha=1$, $SNR=15$dB. We also report the similar results with other parameter settings. From Figure~\ref{fig:poa1} and recall the results on the the performance anormaly in 802.11 WLANs where a user operating on low data rate can drastically degrade the overall network performance, we report the following desirable properties of the adapted TFT strategy at the rate adaptation level:
\begin{compactitem}
\item The PoA is close to $1$ in the studied cases, signifying that the adapted TFT strategy can actually lead the selfish users to a reasonable efficient equilibrium, which is shown to be unique, despite the user selfishness and the lack of coordination among users.
\item The scalability of the adapted TFT strategy is also demonstrated in that the PoA decreases with $N$ and tends to around $1.02$.
\end{compactitem}

We conclude this section with a discussion on the distributed implementation of the best-response strategy to converge to the unique equilibrium.
To compute the best-response strategy, each user should locally compute~\eqref{eq:a2}. To that end, each player can estimate the SNR by adding a feedback in the ACK packets indicating the power at the receiver and compare it with the emission power to further derive $e_i(R_i)$ and $G_i(R_i)$. Besides, each player can calculate $\prod_{j\in{\cal N},j\ne i}\rho_j$ based on local observation by using the approach proposed in \cite{Heusse05}, \cite{Chen09} based on observing the average number of consecutive idle slots between two transmission attempts. Since each user knows its own strategy, he can compute $A_i(R_i)$ in \eqref{eq:a2} locally. Hence, the NE can be achieved based on solely local observation.

\section{Equilibrium Analysis of Higher-level medium access Game and the Global Game}

In this section, we investigate the higher-level medium access game under the adapted TFT framework given that once players choose their CW strategy, they play the lower-level rate adaptation game by operating on the resulting unique equilibrium of the game. In the following analysis, we first show the existence of equilibrium for the medium access game, which is also the system equilibrium. Since there may exist multiple equilibria, we introduce the equilibrium refinement process to find the most efficient one. We then develop a distributed algorithm to converge to the refined equilibrium. We also study the efficiency of this equilibrium.

Before delving into the equilibrium analysis, the following lemma studies a fundamental property by showing that given any strategy of other players, any player $i$ is always better off by decreasing its CW value $W_i$.

\begin{lemma}
Under the non-atomic game assumption, given any strategy of other player, any player $i$ gets higher utility by decreasing its CW value $W_i$.
\label{lemma:less_cw}
\end{lemma}

\begin{proof}
Given any $W_{-i}$ and any $W_i<W_i'$, we need to show that $\widehat{S}_i(W_i,W_{-i})>\widehat{S}_i(W_i',W_{-i})$. To this end, let $(q_j,\tau_j,R_j)_{j\in{\cal N}}$ denote the system state at $(W_i,W_{-i})$ with $(R_j)_{j\in{\cal N}}$ being the equilibrium of the lower-level rate adaptation game and let $(q_j',\tau_j',R_j')_{j\in{\cal N}}$ denote the system state at $(W_i',W_{-i})$ with $(R_j')_{j\in{\cal N}}$ being the equilibrium of the lower-level rate adaptation game, it follows from the assumption of the non-atomic game that the impact of changing $W_i$ and $R_i$ on the strategy and state of other players is negligible. Hence $(q_j,\tau_j,R_j)_{j\in{\cal N}, j\ne i}=(q_j,\tau_j,R_j)_{j\in{\cal N}, j\ne i}$.

Now let $i$ operate on data rate $R_i'$ at $(W_i,W_{-i})$, it follows from~\eqref{eq:q} and~\eqref{eq:p} that $q_i<q_i'$, further leading to $\tau_i>\tau_i'$. Noticing that $\tau_j'=\tau_j, \forall j\in{\cal N}, j\ne i$, we have $\widehat{S}_i(W_i',W_{-i})<S_i(W_i, R_i')$. Recall the definition that given the CW setting $(W_i',W_{-i})$, $R_i$ is the best response to $R_{-i}$, it holds that $\widehat{S}_i(W_i,W_{-i})\le S_i(W_i, R_i')$. Therefore, we have $\widehat{S}_i(W_i,W_{-i})>\widehat{S}_i(W_i',W_{-i})$.
\end{proof}

Lemma~\ref{lemma:less_cw}, whose proof is detailed in xx, justifies the proposed adapted TFT strategy as it is motivated by that smaller CW value leads to higher effective throughput.

\subsection{System Equilibrium Analysis}

The proposed adapted TFT framework can ensure that all users operate on the same CW value, i.e., $W_j=W, \forall j\in{\cal N}$. For any player $i\in{\cal N}$, it can be noted that:
\begin{compactitem}
\item When $\mathbf{W}=\mathbf{0}$\footnote{Throughout this section, to make the analysis concise, we use boldface letters $\mathbf{W}$ to denote a $N$-element vector with each element being $W$, i.e., $\mathbf{W}\triangleq (W_j=W)_{j\in{\cal N}}.$}, meaning that any player $j$ transmits its packets without waiting, we have $\widehat{S}_i(\mathbf{W})=0$\footnote{Recall that $\widehat{S}_i$, defined in Definition~\ref{def:medium access_game}, denotes the utility of $i$ for the higher-level medium access game.};
\item When $\mathbf{W}\rightarrow\mathbf{\infty}$, noticing~\eqref{eq:tau} that $\tau_i\le \frac{2}{W+1}\rightarrow 0$, it holds that $\widehat{S}_i(\mathbf{W})\rightarrow 0$.
\end{compactitem}
Hence there exists $0<W^*_i<\infty$ such that $\widehat{S}_i(\mathbf{W_i^*})$ is maximized.

\begin{definition}
For any player $i\in{\cal N}$, let $\widehat{{\cal W}}_i$ denote the set containing all CW values $W_i$ such that for any $W\le W_i$, it holds that
$\widehat{S}_i(\mathbf{W})\le \widehat{S}_i(\mathbf{W_i})$.
\label{def:w_i}
\end{definition}

Recall the adapted TFT strategy, Definiton~\ref{def:w_i} indicates:
\begin{compactitem}
\item when operating at any $W_i\in\widehat{{\cal W}}_i$, player $i$ has no incentive to decrease its CW value;
\item when operating at some $W_i'\notin\widehat{{\cal W}}_i$, player $i$ can always find $W_i<W_i'$ such that $\widehat{S}_i(\mathbf{W_i})>\widehat{S}_i(\mathbf{W_i'})$, i.e., $i$ is always better off by deviating from $W_i'$ to $W_i$.
\end{compactitem}

\begin{example}
Let us see an example to clarify the above definition and notation. Consider a scenario where $\widehat{S}_i(\mathbf{0})=0$, $\widehat{S}_i(\mathbf{1})=0.1$, $\widehat{S}_i(\mathbf{2})=0.15$, $\widehat{S}_i(\mathbf{3})=0.14$, $\widehat{S}_i(\mathbf{4})=0.16$ and $\widehat{S}_i(\mathbf{W})<\widehat{S}_i(\mathbf{4})$, for any $W\ge 5$. In this example, $\widehat{{\cal W}}_i=\{0,1,2,4\}$. When operating on $W=3$ and $W\ge 5$, player $i$ can increase its utility by decreasing $W$ to $2$ and $4$, respectively (note that under the adapted TFT strategy, other players will also decrease their CW values correspondingly).
\end{example}

Since $\widehat{S}_i(\mathbf{0})=0$ and $\widehat{S}_i(\mathbf{W})\ge 0$ for $W\ge 0$, we have $0\in\widehat{{\cal W}}_i$ and $1\in\widehat{{\cal W}}_i$. Recall that $W_i^*$ is the maximizer of $S_i(\mathbf{W_i^*})$, it follows from Definition~\ref{def:w_i} that $W_i^*\in \widehat{{\cal W}}_i$. Let $\mathcal{\widehat{W}}=\bigcap_{i\in{\cal N}} \widehat{{\cal W}}_i$, $\mathcal{\widehat{W}}$ is not empty by containing at least $\mathbf{0,1}$. We next show in Theorem~\ref{theorem:ne_existence_medium access} that any strategy profile $\mathbf{W}$ with $W\in \mathcal{\widehat{W}}$ is an equilibrium of the higher-level medium access game $G_M$ and that any other strategy profile cannot be an equilibrium.

\begin{theorem}
Under the adapted TFT framework, any strategy profile $\mathbf{W}$ with $W\in \mathcal{\widehat{W}}$ is an equilibrium of $G_M$. Any other strategy profile cannot be an equilibrium of $G_M$.
\label{theorem:ne_existence_medium access}
\end{theorem}

\begin{proof}
We first show that any player $i$ has no incentive to deviate from $\mathbf{W}$ where $W\in \mathcal{\widehat{W}}$.

First, it follows straightforwardly from Lemma~\ref{lemma:less_cw} that $i$ has no incentive to increase $W_i$ from $W$. We now show that $i$ has no incentive to decrease $W_i$ from $W$, either. To this end, recall the adapted TFT framework, decreasing $W_i$ from $W$ to $W'<W$ leads to other players decrease their CW value to $W'$. The system thus operates on $\mathbf{W'}$. Recall Definition~\ref{def:w_i}, we have $\widehat{S}_i(\mathbf{W'})<\widehat{S}_i(\mathbf{W})$. Player $i$ thus has no incentive to decrease $W_i$ from $W$, either. It then follows from the definition of equilibrium that $\mathbf{W}$ is an equilibrium of $G_M$.

We then show that any other strategy profile cannot be an equilibrium of $G_M$. Otherwise, if there exists an equilibrium $\mathbf{W'}$ with $W'\notin{\mathcal{\widehat{W}}}$, then we must have
\begin{eqnarray*}
\widehat{S}_i(\mathbf{W})\le \widehat{S}_i(\mathbf{W'}), \ \forall W\le W', \ \forall i\in{\cal N}.
\end{eqnarray*}
It then holds that $W'\in{\mathcal{\widehat{W}}}$, which leads to contradiction with $W'\notin{\mathcal{\widehat{W}}}$, which concludes our proof.
\end{proof}

\begin{example}
\label{ex:ex2}
To clarify the analysis of Theorem~\ref{theorem:ne_existence_medium access} on the equilibria of $G_M$. Consider an example of $G_M$ with two players with the following utility setting:
\begin{compactitem}
\item $\widehat{S}_1(\mathbf{0})=0$, $\widehat{S}_1(\mathbf{1})=0.1$, $\widehat{S}_1(\mathbf{2})=0.15$, $\widehat{S}_1(\mathbf{3})=0.14$, $\widehat{S}_1(\mathbf{4})=0.16$ and $\widehat{S}_1(\mathbf{W})<\widehat{S}_i(\mathbf{4})$, for any $W\ge 5$;
\item $\widehat{S}_2(\mathbf{0})=0$, $\widehat{S}_2(\mathbf{1})=0.11$, $\widehat{S}_2(\mathbf{2})=0.16$, $\widehat{S}_2(\mathbf{3})=0.165$, $\widehat{S}_2(\mathbf{4})=0.163$, $\widehat{S}_2(\mathbf{5})=0.166$ and $\widehat{S}_2(\mathbf{W})<\widehat{S}_i(\mathbf{5})$, for any $W\ge 6$;
\end{compactitem}
In this example, it can be checked that $\mathcal{W}_1=\{0,1,2,4\}$ and $\mathcal{W}_2=\{0,1,2,3,5\}$; hence $\mathcal{\widehat{W}}=\{0,1,2\}$. Theorem~\ref{theorem:ne_existence_medium access} shows that    there are $3$ equilibria in this example: $\mathbf{0}$, $\mathbf{1}$ and $\mathbf{2}$. It can be easily verified that no player has incentive to deviate from these equilibria under the adapted TFT strategy and that any other strategy profile is not an equilibrium.
\end{example}

The following theorem on the equilibrium of the global joint medium access and rate control game holds naturally.

\begin{theorem}
Under the adapted TFT framework, any strategy profile $(\mathbf{W, R^*(W)})$ with $W\in \mathcal{\widehat{W}}$ and $\mathbf{R^*(W)}$ being the unique equilibrium of $G_R(\mathbf{W})$ is an equilibrium of the global joint medium access and rate control game $G$. Any other strategy profile cannot be an equilibrium of $G$.
\label{theorem:ne_global}
\end{theorem}

\begin{proof}
We first show that any player $i$ has no incentive to deviate from $W^*$. Consider that $i$ deviates from $(W_i^*, R_i^*(\mathbf{W}))$ to $(W_i', R_i')$ with $W_i\ne W_i^*$. Recall the non-atomic game assumption, by operating on $R_i'$ and decreasing from $W_i'$ to $W_i^*$, $\tau_i$ increases and thus $\rho_i$ decreases while $\rho_j$ remains the same for $j\ne i$. Hence $S_i$ increases by switching  from $(W_i', R_i')$ to $(W_i^*, R_i')$. Therefore, player $i$ has no incentive to increase $W_i$. Moreover, player $i$ has no incentive to decrease $W_i$ to some $W_i''<W^*$ since under the adapted TFT framework, this will push other players to $W_i''$ and the system will converge to $(\mathbf{W_i'', R^*(W_i'')})$, where $\mathbf{R^*(W_i'')}$ is the unique equilibrium of $G_M(\mathbf{W_i''})$. As shown in Theorem~\ref{theorem:ne_existence_medium access}, the throughput of $i$ decreases.

Moreover, since $\mathbf{R^*(W)}$ is the unique equilibrium of $G_R(\mathbf{W})$, any player $i$ has no incentive to unilaterally change $R_i^*$, either.

We then show that any other strategy profile cannot be a system equilibrium. Firstly, under the adapted TFT strategy, at any equilibrium, the CW values are identical for all players, otherwise it follows from Lemma~\ref{lemma:less_cw} that the players with higher CW values have incentive to decrease their CW values. Second, the data rate profile must be the unique equilibrium of $G_R$ under the correspondent CW setting. Given the above observation, it holds that the CW setting of any system equilibrium must be an equilibrium of $G_M$, otherwise if $\mathbf{(W,R^*(W))}$ with $W>W^*$ is a system equilibrium, then for any player $i$, by deviating from $W$ to $W^*$, the system will be dragged to $(\mathbf{W^*, R^*(W^*)})$, where it enjoys a higher throughput. It then follows from Theorem~\ref{theorem:ne_existence_medium access} that any strategy profile with $W>W^*$ cannot be a system equilibrium.
\end{proof}

Since $\mathbf{0},\mathbf{1}\in\mathcal{\widehat{W}}$, Theorem~\ref{theorem:ne_global} shows that there always exists at least two equilibria in $G_M$. Generally speaking, among multiple equilibria, some are not desirable from the system's perspective. This can be illustrated by reexamining Example~\ref{ex:ex2} where among the three equilibria, $\mathbf{2}$ is the most efficient one while $\mathbf{0}$ corresponds to system collapse. To address this challenge, a natural method is to remove those less efficient equilibria to achieve a desirable outcome. This is achieved by the \textit{equilibrium refinement}, explored by the following subsection to find the most favorable equilibrium at the system's perspective and to approach the refined equilibrium.

\subsection{Equilibrium Refinement}

We introduce three criteria in the equilibrium refinement process: \emph{fairness}, \emph{Pareto optimality} and \emph{system efficiency}.

\textit{Fairness:} an equilibrium is fair if the system resource is allocated fairly among users. In $G$, all the equilibria $(\mathbf{W, R^*(W)})$ achieve fairness among players in the sense that the players converge to the same CW value and they occupy the channel for the same amont of time at any equilibrium. Such user fairness is inherently enforced by the proposed adapted TFT framework.

\textit{Pareto optimality:} an equilibrium is Pareto optimal if we cannot find another equilibrium where the utility of every player is higher. Let $W^*\triangleq \max_{W\in\mathcal{\widehat{W}}} W$, it holds that only $(\mathbf{W^*, R^*(W^*)})$ is guaranteed to be Pareto optimal because from Definition~\ref{def:w_i}, we have
    \begin{eqnarray}
    S_i(\mathbf{W, R^*(W)})\le S_i(\mathbf{W^*, R^*(W^*)}), \ W<W^*, \ \forall i\in{\cal N}.
    \label{eq:ne_refinement}
    \end{eqnarray}
furthermore, $(\mathbf{W^*, R^*(W^*)})$ is the only Pareto optimum equilibrium among the equilibria if for at least one play $i$, the strict inequality holds in~\eqref{eq:ne_refinement}.

\textit{System efficiency:} following the same analysis on Pareto optimality, we can show that $(\mathbf{W^*, R^*(W^*)})$ achieves the maximal system throughput among the equilibria. Moreover, if for at least one play $i$, the strict inequality holds in~\eqref{eq:ne_refinement}, then the system effective throughput at $(\mathbf{W^*, R^*(W^*)})$ is strictly higher than any other equilibrium.

The equilibrium refinement thus leads to a unique efficient NE $(\mathbf{W^*, R^*(W^*)})$. Next we study how to approach the refined equilibrium. We start by establishing the following lemma that bounds $W^*$ and leads to more efficient search.

\begin{lemma}
Let $G_i^{max}\triangleq \max_{R_i}G_i(R_i)$. Given any strategy profile $\mathbf{W_0}$, denote $x^*$ the root of the equation  $\frac{(1-x)^{n-1}x}{1-(1-x)^{n}}=\frac{\widehat{S}_i(\mathbf{W_0})}{G_i^{max}}$, it holds that $\widehat{S}_i(\mathbf{W})\le \widehat{S}_i(\mathbf{W_0})$ for $W\ge \frac{2}{x^*}$.
\label{lemma:ne_search}
\end{lemma}

\begin{proof}
Recall the formula of $S_i$, noticing that
\begin{eqnarray*}
(1-\rho_i)\prod_{j\in{\cal N},j\ne i}\rho_j=\prod_{j\in{\cal N},j\ne i}\rho_j-\prod_{j\in{\cal N}}\rho_j\le 1-\prod_{j\in{\cal N}}\rho_j,
\end{eqnarray*}
it holds that
\begin{eqnarray*}
\widehat{S}_i(\mathbf{W})=\frac{(1-\rho_i)\prod_{j\in{\cal N},j\ne i}\rho_jG_i(R_i^*(\mathbf{W}))}{1-\prod_{j\in{\cal N}}\rho_j} \\
\le G_i(R_i^*(\mathbf{W}))\le G_i^{max}, \ \forall \ W\ge 0.
\end{eqnarray*}
On the other hand, it can be checked that $T(x)\triangleq \frac{(1-x)^{n-1}x}{1-(1-x)^{n}}$ is monotonously decreasing in $x$ and $T(0)=1$, $T(\infty)=0$. $T(x)=\frac{\widehat{S}_i(\mathbf{W_0})}{G_i^{max}}$ admits a unique solution $x^*$.

When all players operate on $W$, it holds that $\rho_j=1-\tau_j\ge \frac{2}{W}, \forall j\in{\cal N}$. Therefore, $W\ge \frac{2}{x^*}$, it holds that
\begin{eqnarray*}
\widehat{S}_i(\mathbf{W_0})\ge G_i^{max}T\left(\frac{2}{W}\right)\ge T(x^*)\ge \widehat{S}_i(\mathbf{W}),
\end{eqnarray*}
which completes our proof.
\end{proof}

Lemma~\ref{lemma:ne_search} implies that operating on CW larger than $\frac{2}{x^*}$ cannot be a system equilibrium.
Consequently, when searching the efficient equilibrium $\mathbf{W^*}$ derived from the refinement process, if the CW value of players is currently $W$, it suffices to search until $\min_{i\in{\cal N}} W_i$ where $W_i=\frac{2}{x^*_i}$ with $x^*_i$ being the root of $\frac{(1-x)^{n-1}x}{1-(1-x)^{n}}=\frac{\widehat{S}_i(W)}{G_i^{max}}$. Based on this result, we develop a distributed algorithm (Algorithm~1) to cooperatively search and converge to $\mathbf{W^*}$.

The core idea of the first loop in the algorithm is to have a coordinator to synchronize the CW values of players. The coordinator can be any player. Then each player $i$ can construct $\widehat{{\cal W}}_i$ based on Definition~\ref{def:w_i}. The construction process terminates at $\min_{i\in{\cal N}} W_i^{max}$ since from Lemma~\ref{lemma:ne_search}, it suffices to parse CW values until $\min_{i\in{\cal N}} W_i^{max}$, which is dynamically updated in the algorithm. In the second loop of the algorithm, by applying the adapted TFT strategy, the system converges gradually to the efficient equilibrium $W^*$.

As can be noted from Algorithm~1, it can be implemented in a distributed fashion and any player can be the coordinator to synchronize the CW values. Moreover, all players have incentive to participate the cooperative search of the efficient equilibrium. This can be shown as follows: it follows from the adapted TFT strategy that the system will operate on the same CW value for all players; hence the system will operate on an equilibrium, otherwise we can find a player who can increase its utility by decreasing its CW value; thus the CW value is monotonously decreasing and will reach an equilibrium because the smallest CW value $\mathbf{0}$ is an equilibrium; since $\mathbf{W^*}$ is the most efficient equilibrium among the potentially multiple system equilibria, any selfish but rational player has incentive to operate on $W^*$ by participating the cooperative search.


A note on the robustness of the algorithm. We observe via numerical experiments (detailed in Section~\ref{sec:simu}) that by deviating from $W^*$ (e.g., due to variation in measurement or a large step size instead of $1$) and by varying the number of players, the user can still achieve reasonably efficient point, with at least 80\% of the throughput achieved at the efficient equilibrium. This observation demonstrates the robustness of the algorithm in dynamic scenarios with frequent arrival and departure. Such robust feature can significantly facilitate the practical implementation of the algorithm.

\begin{algorithm}[htbp]
\caption{Searching the efficient equilibrium $\mathbf{W^*}$: executed at each player $i$}
\label{algo:efficient_ne_search}
\begin{algorithmic}[1]
\STATE \textbf{Initialization:} Set $W_i=0$, $\widehat{{\cal W}}_i=\{0\}$, $S_i^{max}=0$, $W_i^{max}=\infty$ and set $\epsilon$ to a small value
\STATE \textbf{Start:} Any player $l$ broadcasts a message \texttt{StartSearch} to start the searching process
\LOOP
\IF{a message \texttt{SearchStop} received}
    \STATE Quit the loop by going to line 25
\ENDIF
\IF{a message \texttt{StartSearch} or \texttt{IncreaseW} received \textbf{or} a message \texttt{StartSearch} or \texttt{IncreaseW} sent in case $i=l$}
    \STATE Increase $W_i$ by $1$ and wait a short period of time for others to increase their CW values
    \STATE Count the number of acknowledged packets $n_s$ during a period $t_m$ and measure the average effective throughput $S_i(W_i)=\frac{n_sR_iT}{t_m}$
    \IF{$S_i>S_i^{max}$}
        \STATE Set $S_i^{max}=S_i$, $W_i^*=W_i$, $W_i^{max}=\frac{2}{x^*}$ where $x^*$ is the root of the equation $\frac{(1-x)^{n-1}x}{1-(1-x)^{n}}=\frac{S_i}{G_i^{max}}$
        \STATE Add $W_i$ into $\widehat{{\cal W}}_i$
    \ENDIF
    \IF{$W_i\ge W_i^{max}$}
        \STATE Set $W_i=W_i^*$
        \STATE Send a message \texttt{SearchStop} containing $W_i^*$
        \STATE Quit the loop by going to line 25
    \ELSE
        \IF{$i=l$}
            \STATE Wait sufficient long time for the other players to finish the above operations
            \STATE Send a message \texttt{IncreaseW}
        \ENDIF
    \STATE Go to \textbf{loop}
    \ENDIF
\ENDIF
\ENDLOOP
\REPEAT
\STATE Measure CW values of others during a period of time
\IF{$\min_j W_j< W_i-\epsilon$}
    \STATE Set $W_i$ to the largest element in $\widehat{{\cal W}}_i$ that is smaller than $W_i$
\ENDIF
\UNTIL{$\min_j W_j\ge W_i-\epsilon$}
\end{algorithmic}
\end{algorithm}

\subsection{Efficiency of System Equilibrium}

In this subsection, we investigate the efficiency of the system equilibrium $(\mathbf{W^*, R^*(W^*)})$ derived previously. To make our analysis tractable, we focus here on the symmetrical scenario where the channel conditions are the same for all players. Nevertheless, we observe similar results for generic scenarios, as detailed in the simulation results in Section~\ref{sec:simu}.

In the symmetrical scenario, it follows from the analysis of the previous two subsections that $W^*=W_i^*, \forall i\in{\cal N}$. This means that the refined equilibrium $\mathbf{W^*}$ is composed of the maximizers of all the individual utility function $\widehat{S}_i$. Let $\mathbf{(W^{opt}, R^{opt})}$ denote the system optimum, since $\mathbf{W^*}$ is the refined equilibrium of $G_M$, we have
$$\widehat{S}_i(\mathbf{W^*, R^*(W^*)})\ge \widehat{S}_i(\mathbf{W^{opt}, R^*(W^{opt})}).$$
Noticing the symmetry of players, we have
\begin{eqnarray*}
\frac{\sum_{i\in{\cal N}}S_i(\mathbf{W^{opt}, R^{opt}})}{\sum_{i\in{\cal N}}\widehat{S}_i(\mathbf{W^*, R^*(W^*)})} = \frac{S_i(\mathbf{W^{opt}, R^{opt}})}{\widehat{S}_i(\mathbf{W^*, R^*(W^*)})} \\
\le \frac{S_i(\mathbf{W^{opt}, R^{opt}})}{\widehat{S}_i(\mathbf{W^{opt}, R^*(W^{opt})})},
\end{eqnarray*}
which is the PoA of $G_R(\mathbf{W^{opt}})$. This result readily indicates that the system level PoA equals to the PoA of the lower-level rate adaptation game at $\mathbf{W^{opt}}$. Since we have demonstrated that the PoA of the lower-level rate adaptation game is very close to $1$, we hence have that the refined equilibrium of the global game is also very close to system optimum from a social perspective.

\section{NUMERICAL EXPERIMENTS}
\label{sec:simu}

In this section, we present a suite of numerical experiments to evaluate the proposed adapted TFT strategy by demonstrating and validating some of the theoretical results of the joint medium access and rate adaptation game studied in the paper, especially for the cases that we are not able to investigate analytically. Specifically, we focus on several scenarios indicative of the typical interactions among the players in the game, starting with the symmetrical case with homogeneous players, continuing with a more sophisticated scenario with two classes of homogenous players, and finally considering the generic asymmetrical scenario with heterogeneous players randomly parameterized. In particular, we investigate the structure of equilibrium of these scenarios and compare it to the system optimum.

\subsection{Symmetrical Scenario}

We start by analyzing the symmetrical scenario with homogeneous players. To this end, we simulate a standard 802.11 WLAN of $N$ homogeneous users operating on $64$ QAM with $\alpha=1$, $SNR=15$dB (please refer to~\eqref{eq:qam}). We set a large rate range with $R_{min,i}=1$Mbps and $R_{max,i}=100$Mbps. Table~\ref{table;ne_sym} and Figure~\ref{fig:ne_sym} compare the efficient system equilibrium found by applying Algorithm~1 to the global optimum. By comparing the players' equilibrium strategy and the optimal strategy, we observe that players are slightly more aggressive by using smaller CW value and lower data rate leading to less transmission error. However, as suggested by the simulation results, the efficiency loss of the system due to players' selfishness is very small, which demonstrates that the proposed adapted TFT framework can bring about a reasonably efficient equilibrium with only a small system utility loss. This result is especially meaningful given the result on the performance anormaly of multi-rate 802.11 WLAN where users tend to operate on low data rate~\cite{Heusse03}.

Figure~\ref{fig:around_ne} further studies the robustness of Algorithm~1 searching the efficient system equilibrium. To this end, we focus on the case $N=12$ where the efficient system equilibrium is $W^*=220$ and we study the case where the system does not operate exactly on the equilibrium: (1) the system operates on the CW around $W^*$ varying from $100$ to $500$, this may be due to the choice of a large step size in the algorithm; (2) there are users departing and /or entering without rerunning the algorithm, this is to simulate the case where the algorithm is not run too frequently, we thus study the system efficiency by varying the number of players. We report from the result of Figure~\ref{fig:around_ne} that even in the case where the system cannot operate on the exact equilibrium, the system can still achieve at least $85$\% the optimal utility, which demonstrates the robustness of Algorithm~1. This robust and
tolerant feature can significantly facilitate the implementation of the algorithm.

\begin{table}[ht]  \centering
\begin{tabular}{|l|llllll|}
\hline
$N$ & 2 & 7 & 12 & 17 & 22 & 27 \\
\hline
$W^{NE}$ & 45 & 110 & 220 & 345 & 510 & 710 \\
$W^{opt}$ & 55 & 140 & 265 & 430 & 655 & 905 \\
\hline
\end{tabular}
\caption{System equilibrium vs. global optimum}
\label{table;ne_sym}
\end{table}

\begin{figure}[htbp]
\centering
\includegraphics[width=4.5cm, angle=270]{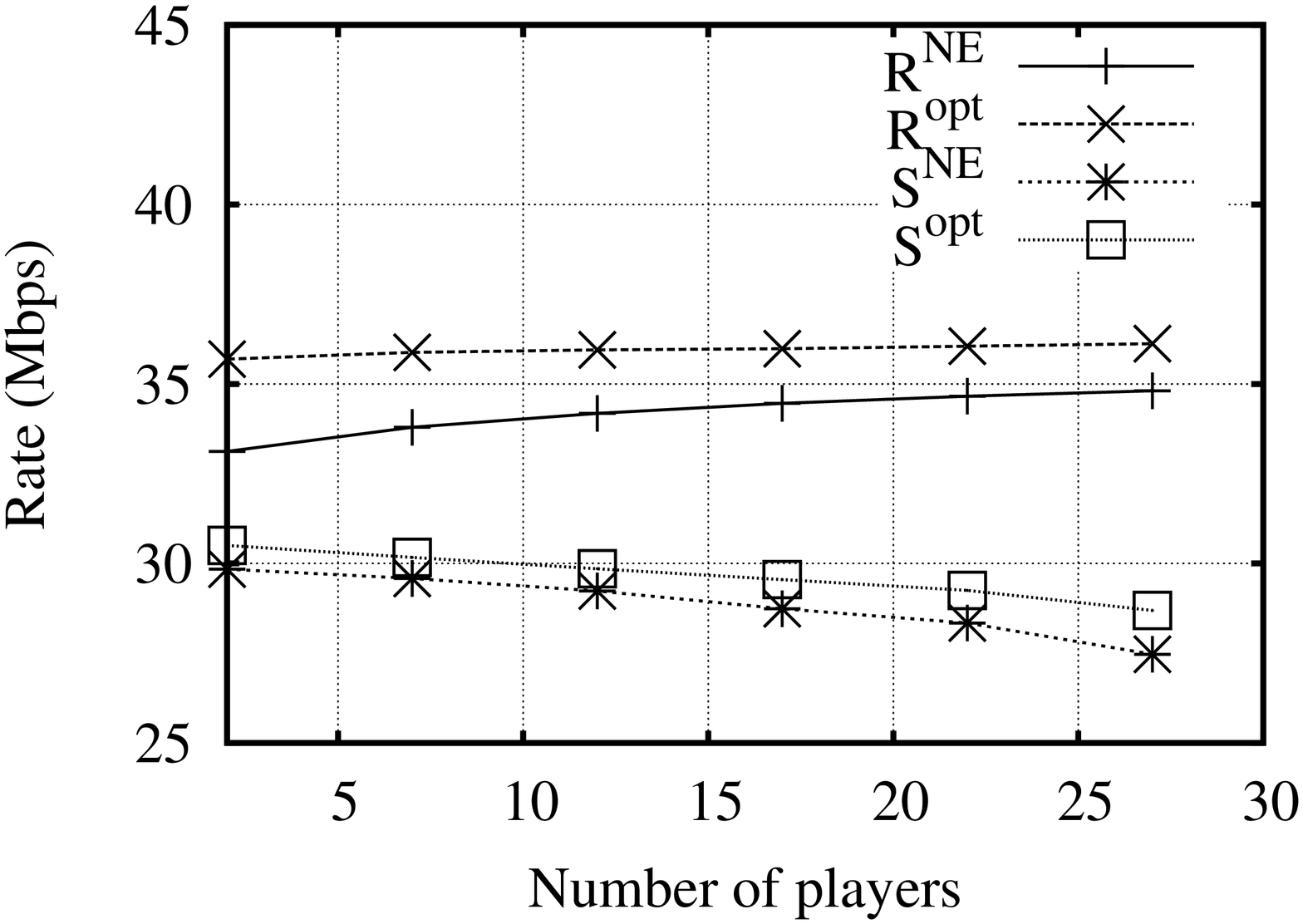}
\caption{System equilibrium vs. global optimum}
\label{fig:ne_sym}
\end{figure}

\begin{figure}[htbp]
\centering
\includegraphics[width=4.5cm, angle=270]{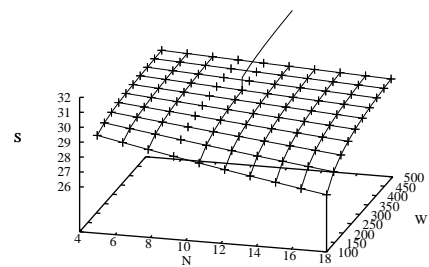}
\caption{Throughput around equilibrium: robustness analysis of Algorithm~1}
\label{fig:around_ne}
\end{figure}

\subsection{Scenario with two Classes of Players}

We proceed to study a more sophisticated scenario composed of two classes, namely L and H, of homogeneous players whose SNR are $SNR_L=10$ and $SNR_H=20$, respectively. The other parameters are set to the same values as before. The results, shown in Figure~\ref{fig:ne_hyb} and Table~\ref{table:ne_hyb}, demonstrate again that although slightly more aggressive due to individual selfishness, the efficient system equilibrium is very close to the system optimum from a social perspective.

\begin{figure}[htbp]
\centering
\includegraphics[width=4.5cm, angle=270]{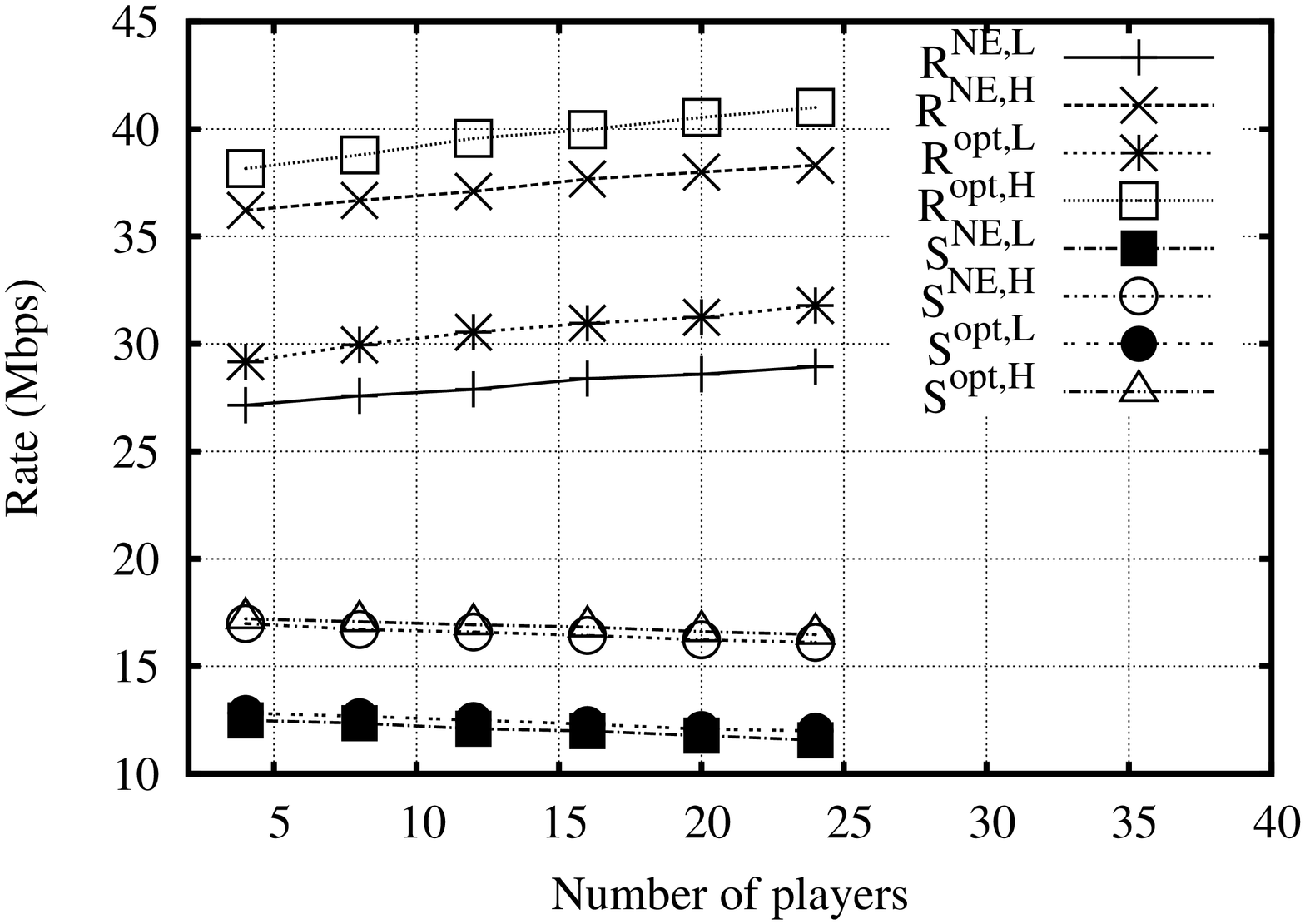}
\caption{System equilibrium vs. global optimum}
\label{fig:ne_hyb}
\end{figure}

\begin{table}[ht]  \centering
\begin{tabular}{|l|llllll|}
\hline
$N$ & 4 & 8 & 12 & 16 & 20 & 24 \\
\hline
$W^{NE}$ & 70 & 130 & 190 & 335 & 550 & 810 \\
$W^{opt}$ & 80 & 160 & 285 & 470 & 665 & 900 \\
\hline
\end{tabular}
\caption{System equilibrium vs. global optimum}
\label{table:ne_hyb}
\end{table}

\subsection{Asymmetrical Scenario}

We finally consider various heterogeneous scenarios with asymmetrical players. More specifically, we simulate a WLAN of $N$ players, each with a SNR randomly chosen from $[5,25]$. For each $N$, we run $100$ simulations with random SNR and plot the average Price of Anarchy in Figure~\ref{fig:poa_asym}. In the simulation, we make the following observations: (1) the convergence to a system equilibrium is always achieved at the rate adaptation level under a given CW value, which confirms our theoretical results; (2) the PoA remains small, with the average value bounded by $1.12$ and not exceeding $1.25$ in any simulated case, which demonstrates  the efficiency of the proposed adapted TFT strategy in generic scenarios; (3) the PoA is decreasing in the number of players which shows the good scalability of the proposed mechanism.

\begin{figure}[htbp]
\centering
\includegraphics[width=4.5cm, angle=270]{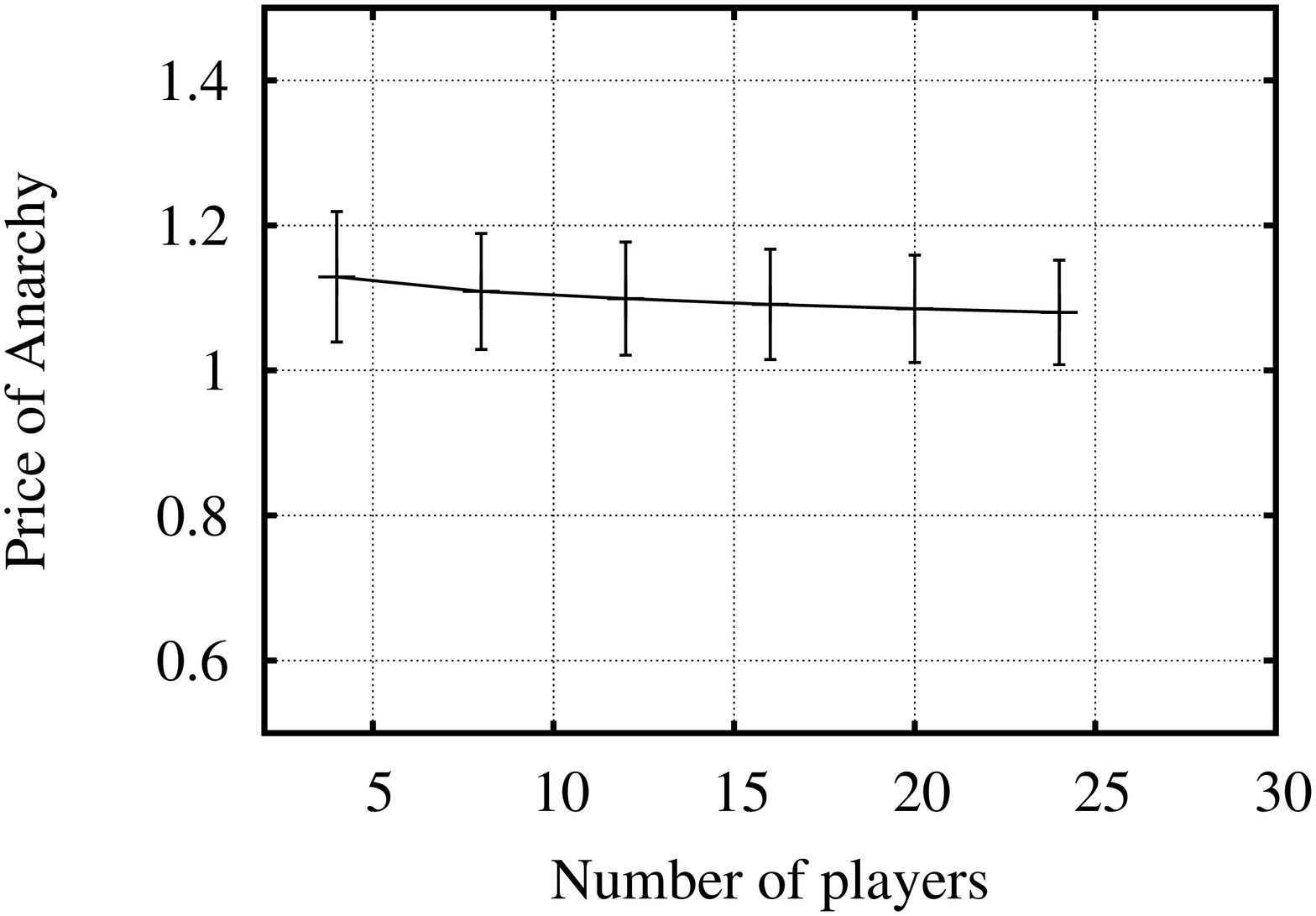}
\caption{System equilibrium vs. global optimum}
\label{fig:poa_asym}
\end{figure}

\section{Conclusion}
\label{sec:conclu}

In this paper, we have investigated the joint rate adaptation and medium control in WLANs from a non-cooperative game theoretic perspective. We have developed an adapted TFT strategy to orient the network to an efficient equilibrium where users can jointly configure their CW size and data rate selfishly. We have formulated the interactions among selfish users under the adapted TFT framework as a non-cooperative joint medium access and rate adaptation game. By analyzing the structural properties of the game, we have provided insights on the interaction between rate adaptation and 802.11 medium access in competitive setting. We have shown that the game has multiple equilibria, which, after the equilibrium refinement process that we develop, reduce to a unique efficient one. We have developed a distributed algorithm to achieve this equilibrium and demonstrated that the equilibrium achieves the performance very close to the system optimum from a social perspective.

\bibliographystyle{unsrt}

{
\bibliography{reference}}

\end{document}